\documentclass[a4paper,11pt]{article}
\pdfoutput=1

\usepackage[margin=1in]{geometry}
\usepackage[utf8]{inputenc}
\usepackage[T1]{fontenc}
\usepackage{lmodern}
\usepackage[english]{babel}
\usepackage{microtype}
\usepackage{amsmath,amssymb,amsfonts,amsthm,mathtools,bm}
\usepackage{booktabs,tabularx,array,multirow}
\usepackage{graphicx,subcaption,float}
\usepackage{placeins}
\usepackage[inline]{enumitem}
\usepackage{xcolor}
\usepackage{url}
\usepackage[pdfencoding=auto,psdextra,hidelinks,colorlinks=true,
 linkcolor=blue,citecolor=blue,urlcolor=blue,pagebackref=true]{hyperref}
\usepackage{bookmark}
\usepackage[nameinlink,noabbrev]{cleveref}
\newfloat{algorithm}{tbp}{loa}
\floatname{algorithm}{Algorithm}
\crefname{algorithm}{algorithm}{algorithms}
\Crefname{algorithm}{Algorithm}{Algorithms}

\newcommand{\citep}[1]{\cite{#1}}
\makeatletter
\newcommand{\citet}[1]{%
  \ifcsname citetname@#1\endcsname
    \csname citetname@#1\endcsname~\cite{#1}%
  \else
    \cite{#1}%
  \fi}
\expandafter\def\csname citetname@KG2011\endcsname{Kamenica and Gentzkow}
\expandafter\def\csname citetname@Weissman2003\endcsname{Weissman et al.}
\makeatother

\newcommand{\Om}{\Omega}
\newcommand{\A}{\mathcal A}
\newcommand{\M}{\mathcal M}
\newcommand{\Tset}{\Theta}

\newcommand{\E}{\mathbb E}
\newcommand{\1}{\mathbf 1}
\newcommand{\calF}{\mathcal F}
\newcommand{\calX}{\mathcal X}
\newcommand{\calC}{\mathcal C}

\newcommand{\calH}{\mathcal H}

\DeclareMathOperator{\supp}{supp}
\DeclareMathOperator*{\argmax}{arg\,max}

\newcolumntype{Y}{>{\raggedright\arraybackslash}X}

\newtheorem{theorem}{Theorem}[section]
\newtheorem{proposition}[theorem]{Proposition}
\newtheorem{lemma}[theorem]{Lemma}
\newtheorem{corollary}[theorem]{Corollary}
\newtheorem{assumption}[theorem]{Assumption}
\newtheorem{definition}[theorem]{Definition}

\begin{document}

\title{Certified Learning and Equilibrium Implementation\\
under Opaque Partial Commitment}
\author{Shuyang Zhang and Xiangtian Li\\
Shanghai Artificial Intelligence Laboratory\\
\texttt{zhangshuyang@pjlab.org.cn}}
\maketitle

\begin{abstract}
We extend Bayesian persuasion to an opaque partial-commitment environment in
which a sender is bound by an installed information policy only with
probability $\rho$, the realization of binding is hidden, and the receiver does
not observe the persistent structural environment.
The receiver first sees a payoff-neutral, nonmanipulable calibration sample and
then faces a fresh, non-certified deployment interaction.  At the population
level, the calibration law identifies the receiver-facing reduced form and its
equivalence class in the finite type dictionary.  The latent kernels remain
unidentified exactly when observationally equivalent candidate types contain
different latent pairs; without decomposition restrictions, the same reduced
form can also admit multiple latent representations.  We characterize type-wise
$\rho$-implementability, construct the receiver's Bayesian belief over the
complete hidden node, and prove a static direct-following implementation theorem.
After every calibration history that passes a posterior-predictive obedience
test, the deployment assessment is an exact perfect Bayesian equilibrium:
Bayes consistency, receiver sequential rationality, sender sequential
rationality, and off-path completion are all verified.  Under finite-type
separation, common recommendation support, and a positive obedience margin, the
test activates such an equilibrium with high probability. Our results keep statistical
failure probability distinct from equilibrium approximation.  Finally, we
embed the robust value frontier, support-wise linear-programming
algorithm, and binary-action fractional-knapsack specialization into this
implementation framework.
\end{abstract}

\section{Introduction}
\label{sec:introduction}
Bayesian persuasion begins with a transparent commitment benchmark.  An
informed sender selects and commits to an information structure, the receiver
knows that structure, and Bayes' rule turns each signal into a posterior on
which the receiver best responds~\cite{KG2011,Kamenica2019}.  This transparency
supports both the concavification approach and the linear-programming
formulations used in algorithmic information design
~\cite{BergemannMorris2016,BergemannMorris2019,DughmiXu2016,Dughmi2017}.
It also makes obedience immediate to state: a direct signal can be labeled by
the action it recommends, and the corresponding joint distribution is feasible
only if following every positive-probability recommendation is optimal.

Many recommendation systems have less transparent commitment.  A platform may
use a binding policy in some instances and retain discretion in others, while a
user observes neither the source of the current recommendation nor the latent
policy.  Ex-post states may nevertheless be auditable.  This creates two
separate questions.  First, what receiver-facing object can be learned without
recovering the latent mechanism?  Second, when does such learning support an
equilibrium in the subsequent strategic interaction?  A statistical argument
that the receiver will often follow is not, by itself, an equilibrium
foundation: it does not establish sender incentives, Bayes-consistent beliefs at
the deployment information set, or behavior after zero-probability messages.

We answer these questions in a deliberately institution-assisted model.  A
persistent structural type $\theta$ indexes the sender payoff and an installed
pair of recommendation kernels.  In each deployment, an independent hidden
shock $c$ binds the installed kernel with probability $\rho$; otherwise the
sender chooses a message.  Before deployment, a trusted certifier produces an
i.i.d.\ sample of recommendation--state pairs from the same structural
environment.  The records are authentic, the state is revealed after each
calibration recommendation, and the sender cannot affect the sample.
Calibration is payoff neutral and precedes strategic play.  A fresh state and a
fresh binding shock are drawn only after calibration.

This timing is substantive.  We do \emph{not} describe calibration as an
equilibrium path or as endogenous reputation formation.  It is public pre-play
information that changes the common prior of the deployment continuation.  The
equilibrium object is the one-shot continuation following a realized
calibration record.  If calibration were manipulable or payoff relevant, its
generation and the associated continuation incentives would have to be
incorporated into a dynamic signaling or reputation game; those incentives are
outside the present model.

For each type, the hidden binding and discretionary kernels
$(\pi^\theta,\overline\sigma^\theta)$ induce the receiver-facing reduced form
\[
 x^\theta_{\omega m}
 =\mu(\omega)\left[
 \rho\pi^\theta(m\mid\omega)
 +(1-\rho)\overline\sigma^\theta(m\mid\omega)\right].
\]
The reduced form is the immediate statistical target.  Within the finite type
dictionary, its population law identifies the equivalence class
$[\theta]=\{\vartheta:x^\vartheta=x^\theta\}$; if
$\theta\mapsto x^\theta$ is injective, this class is a singleton and the stored
latent pair is recovered by dictionary lookup.  Calibration updates a posterior
$\lambda_N(\theta\mid Z_N)$ over
structural types.  At deployment, this posterior and the installed strategies
generate a belief
$q_N(\theta,\omega,c\mid m,Z_N)$ over the \emph{complete} hidden node.  Its
state marginal is the posterior-predictive belief
\[
 \overline\beta_N(\omega\mid m,Z_N)
 =
 \frac{\sum_\theta\lambda_N(\theta\mid Z_N)x^\theta_{\omega m}}
 {\sum_\theta\lambda_N(\theta\mid Z_N)p_m(x^\theta)}.
\]
Thus the learning bridge is predictive:
\[
 Z_N\ \Longrightarrow\ \lambda_N\
 \Longrightarrow\ q_N\
 \Longrightarrow\ \overline\beta_N.
\]
The implementation argument does not require the receiver to identify the
realized sender payoff or separate $\pi^\theta$ from
$\overline\sigma^\theta$ when observational equivalence prevents doing so.

Our contributions are fourfold.
\begin{enumerate}[leftmargin=1.8em,label=(\roman*)]
\item \emph{Environment, reduced forms, and identification.}
We separate the persistent type $\theta$, fresh payoff state $\omega$, and
transient binding shock $c$, characterize type-wise $\rho$-implementability,
and state the exact finite-dictionary identification boundary: the population
calibration law identifies $x^\theta$ and $[\theta]$, while a latent coordinate
is identified only if it is constant on that class.  We separately characterize
multiplicity in the unrestricted decomposition space; see
\Cref{sec:environment,sec:reduced-form}.
\item \emph{Equilibrium implementation.}
We construct the receiver's Bayesian belief over the complete hidden node
and prove a posterior-mixture implementation theorem.  A public predictive
obedience certificate then activates an exact PBE after every accepted record,
including sender and receiver sequential rationality, Bayes consistency, and
off-path completion; see
\Cref{sec:certified-learning,sec:equilibrium-implementation}.
\item \emph{Finite-sample activation.}
For a finite separated type family, common recommendation support and a
positive obedience margin turn posterior concentration into a finite-sample
guarantee of high-probability entry into the exact-PBE region.  The statistical
failure probability remains distinct from any equilibrium-approximation
parameter; see \Cref{sec:equilibrium-implementation,sec:robust-statistical}.
\item \emph{Robust design and computation.}
We embed the robust value frontier in the implementation framework, relate the
obedience margin to certification speed, give a support-wise LP test for
positive-margin attainability, and recover a bounded fractional-knapsack
specialization for binary actions; see
\Cref{sec:robust-statistical,sec:computation,sec:computational-evidence}.
\end{enumerate}

The scope is conditional implementation, not endogenous mechanism selection.
The type-contingent protocol is already installed; our theorem shows when its
prescribed deployment behavior forms a continuation PBE after calibration.  A
designer's optimization over reduced forms is therefore a design layer, not
proof that a privately informed sender
voluntarily selects the protocol ex ante.  This limitation is important in
interpreting both the contribution and its relation to work that explicitly
models experiment-choice incentives.

The paper proceeds as follows.  \Cref{sec:related-literature} locates the model.
\Cref{sec:environment,sec:reduced-form} define the game and reduced forms.
\Cref{sec:certified-learning} constructs Bayesian and predictive beliefs.
\Cref{sec:equilibrium-implementation} proves static and post-calibration
implementation.  \Cref{sec:robust-statistical,sec:computation} develop the
robust statistical and algorithmic layers,
\Cref{sec:computational-evidence} reports reproducible illustrations, and
\Cref{sec:conclusion} concludes.  Appendix~\ref{app:information-roadmap} gives
the canonical benchmark and information-coordinate roadmap; institutional
scope and substantive assumptions are collected in
Appendix~\ref{sec:discussion}, followed by supporting tables and proofs.

\section{Relation to the Literature}
\label{sec:related-literature}
\paragraph{Bayesian persuasion and information design.}
Our reduced-form variables are joint distributions over states and direct
recommendations, as in classical Bayesian persuasion and Bayes-correlated
equilibrium~\cite{KG2011,BergemannMorris2016,BergemannMorris2019}.  The
difference is not a new revelation principle.  The receiver does not initially
know which member of a common structural family generated the recommendation,
and the sender can choose a message in the nonbinding branch.  Consequently,
obedience must be evaluated under a posterior mixture over structural types,
and sender sequential rationality must be verified separately.

\paragraph{Imperfect commitment and weak institutions.}
The hidden binding shock is closest to probabilistic or institutionally weak
commitment.  Min studies a sender whose commitment to an information structure
binds with an exogenous probability~\cite{Min2021}.  Lipnowski, Ravid, and
Shishkin study persuasion through a study that the sender can influence with
some probability~\cite{LipnowskiRavidShishkin2022}.  Lin and Liu instead
restrict undetectable manipulation by requiring the sender to preserve the
message distribution~\cite{LinLiu2024}.  These papers motivate the
economic relevance of imperfect commitment.  Our object differs: we condition
on an installed type-contingent protocol, add certified pre-play learning about
its receiver-facing distribution, and prove a direct-following deployment PBE.
We do not characterize all equilibria of an unrestricted imperfect-commitment
communication game.

\paragraph{Experiment choice and equilibrium selection.}
Perfect Bayesian Persuasion explicitly studies sender incentives in choosing an
experiment and avoids presupposing sender-favorable receiver tie breaking
~\cite{LipnowskiRavidShishkin2025}.  That question is complementary and marks a
boundary of our result.  Our sender incentive theorem concerns message choice
\emph{inside} the installed nonbinding protocol.  It neither derives voluntary
installation nor selects an experiment-choice equilibrium.  A positive
obedience margin eliminates receiver ties in the activated continuation, but
does not by itself solve the ex-ante scheme-selection problem.

\paragraph{Learning and inferable schemes.}
Online persuasion typically studies a sender or designer who repeatedly
chooses signaling schemes and controls regret while learning receiver types,
an unknown prior, or receiver utilities
~\cite{CastiglioniCelliMarchesiGatti2020,Bacchiocchi2024,LinLi2025}.  Our
calibration is instead an exogenous batch sample; the retained prefix-regret
quantity is a virtual stabilization statistic, not strategic Stage-I utility.
Work on
inferable signaling schemes instead asks how a receiver can infer a
full-commitment scheme from repeated observations
~\cite{Probine2025}.  That recent preprint and our model share a predictive
focus, but here the receiver learns a persistent reduced-form
environment from certified pairs; the sender does not control the learning
sample.  Within the finite dictionary, the latent pair is unidentified
precisely when observationally equivalent types carry different latent pairs;
without decomposition restrictions, the aggregate law's nonparametric fiber
may contain a continuum.  Our contribution is the bridge from that predictive
object to a Bayes-consistent continuation equilibrium.

\paragraph{Cheap talk and reputation.}
When $\rho=0$, the model becomes a restricted direct-following implementation:
the prescribed message must be sender best in every state.  It is therefore a
subset of, not a characterization of, cheap-talk equilibria
~\cite{CrawfordSobel1982}; computational characterizations of cheap talk study
a broader endogenous communication problem~\cite{AlgorithmicCheapTalk2023}.
Dynamic reputation has a different
source of discipline.  A persistent private type plays strategically across
histories, and current behavior affects future beliefs and continuation payoffs
~\cite{KrepsMilgromRobertsWilson1982,FudenbergLevine1989}.  Best and Quigley,
for example, study a long-lived sender whose current communication changes a
public accuracy record and hence future credibility~\cite{BestQuigley2024}.
Our binding shock is redrawn at deployment, calibration records are
institutionally generated, and the sender has neither a strategic calibration
action nor a continuation payoff.  The resulting model avoids the recursive
belief and continuation-incentive constraints of a reputation game.

Appendix~\ref{app:information-roadmap} translates the nearest learning and
commitment models into common information coordinates.
\Cref{tab:roadmap-information,tab:roadmap-target} separately compare primitive
knowledge and commitment, dynamic observations, learning targets, and
equilibrium or optimization objects.

\section{The Certified-Calibration Environment}
\label{sec:environment}
\subsection{Primitives and the structural type}

Let $\Om$ be a finite state space, $\A$ a finite receiver action space, and
$\M=\A$ a finite direct-message space.  We write $m\in\M$ for a recommendation
and $a\in\A$ for the receiver's final action.  The interpretation of the direct
label $m$ is ``choose action $m$.''  The payoff-state prior
$\mu\in\Delta(\Om)$ has full support.  Receiver utility
$u_R:\A\times\Om\to[0,1]$ and the binding probability $\rho\in[0,1]$ are common
across structural types.

A persistent structural type $\theta$ is drawn from the finite set $\Tset$
according to a full-support common prior $\lambda_0$.  The following
type dictionary is common knowledge:
\begin{equation}
\theta\longmapsto
\left(
u_S^\theta,\{B_{\omega}^{S,\theta}\}_{\omega\in\Om},
x^\theta,\pi^\theta,\overline\sigma^\theta
\right).
\label{eq:type-dictionary}
\end{equation}
Here $u_S^\theta:\A\times\Om\to[0,1]$ and
\[
B_{\omega}^{S,\theta}
=\argmax_{a\in\A}u_S^\theta(a,\omega)
\]
is the sender-best action set.  The kernel
$\pi^\theta:\Om\to\Delta(\M)$ is the installed binding policy, whereas
$\overline\sigma^\theta:\Om\to\Delta(\M)$ is the installed
\emph{discretionary prescription}.  We use a bar to distinguish this
prescription from the sender's endogenous deployment strategy.  The prescribed
kernel satisfies
\begin{equation}
\supp\overline\sigma^\theta(\cdot\mid\omega)
\subseteq B_{\omega}^{S,\theta}
\quad\text{for every }(\theta,\omega).
\label{eq:sender-best-prescription}
\end{equation}
The receiver-facing joint distribution is
\begin{equation}
x^\theta_{\omega m}
=\mu(\omega)\phi^\theta(m\mid\omega),\qquad
\phi^\theta(m\mid\omega)
=\rho\pi^\theta(m\mid\omega)
{}+(1-\rho)\overline\sigma^\theta(m\mid\omega).
\label{eq:type-reduced-form}
\end{equation}
Thus $x^\theta$ is a distribution on $\Om\times\M$ whose state marginal is
$\mu$.  We write
$p_m(x)=\sum_{\omega}x_{\omega m}$ and, when $p_m(x)>0$,
$\beta_x(\omega\mid m)=x_{\omega m}/p_m(x)$.

\begin{assumption}[Installed and stable type-contingent protocol]
\label{ass:installed}
The map in \eqref{eq:type-dictionary}, $\lambda_0$, $\mu$, $u_R$, and $\rho$
are common knowledge.  The realized $\theta$ remains fixed across calibration
and deployment and is observed by the sender, the certifier, and the binding
device, but not by the receiver.  The pair
$(\pi^\theta,\overline\sigma^\theta)$ is installed before calibration and is
not selected by a strategic move in the game studied here.
\end{assumption}

Assumption~\ref{ass:installed} completes the original ``unknown mechanism''
description as a Bayesian model.  The receiver does not observe the realized
$u_S^\theta$, $B^{S,\theta}$, $x^\theta$, $\pi^\theta$, or
$\overline\sigma^\theta$; she knows the possible values and their common prior.
Candidate types are required to be sender feasible through
\eqref{eq:sender-best-prescription}, but they need \emph{not} be receiver
obedient.  If every candidate reduced form were obedient, posterior mixtures
would be obedient by linearity and calibration would be unnecessary.

\subsection{Stage I: certified calibration}

After $\theta$ is drawn, a trusted certifier generates the public transcript
\begin{equation}
Z_N=((\omega_t,m_t))_{t=1}^N
\sim (x^\theta)^{\otimes N}.
\label{eq:calibration-sample}
\end{equation}
For each observation, the recommendation is recorded and the state is
authentically revealed ex post.  The certifier observes or authenticates the
realized structural type sufficiently to run the installed type-contingent
protocol.  The sender cannot choose or alter calibration messages, states, or
records.

\begin{assumption}[Scope of certification]
\label{ass:certification}
Calibration is nonmanipulable, payoff neutral, and contains no receiver action.
The certifier guarantees data provenance, the stability of the data-generating
reduced form, and equality of the calibration and deployment structural type.
It does \emph{not} certify receiver obedience or membership of
$x^\theta$ in a robust feasible set.
\end{assumption}

The last clause prevents certification from assuming the conclusion.  Stage I
is an exogenous statistical experiment before the game with payoff-relevant
actions.  It contains neither a sender strategy nor a receiver best-response
condition, and we make no PBE claim about its generation.

\subsection{Stage II: non-certified deployment}

After a public transcript $z\in(\Om\times\M)^N$, a public gate either activates
or rejects the communication protocol.  The certificate defining the gate is
introduced in \Cref{sec:certified-learning}.  If the record is rejected, the
communication channel is disabled.  A fresh state $\omega\sim\mu$ is drawn, the
sender has no move, and the receiver chooses a no-information best response
\begin{equation}
a^0\in\argmax_{a\in\A}\sum_{\omega\in\Om}\mu(\omega)u_R(a,\omega).
\label{eq:fallback-action}
\end{equation}
For a feasible rejected record, the receiver retains the Bayesian posterior on
$\theta$, but the fresh state remains independent with law $\mu$; because
$u_R$ is type independent, \eqref{eq:fallback-action} is sequentially optimal.
After an infeasible record, any belief on $\theta$ can be assigned while the
fresh state still has law $\mu$.  The gate is mechanical, and the sender cannot
bypass the disabled channel.  This is a fallback decision problem, not target
implementation.

If the record is accepted, the deployment continuation $\Gamma(z)$ has the
following timing.
\begin{enumerate}[leftmargin=1.8em,label=\arabic*.]
\item Nature independently draws a fresh payoff state $\omega\sim\mu$ and a
transient binding shock $c\sim\operatorname{Bernoulli}(\rho)$.
\item If $c=1$, the type-aware binding device draws
$m\sim\pi^\theta(\cdot\mid\omega)$; the sender has no message choice.
\item If $c=0$, the sender observes $(z,\theta,\omega)$ and chooses a message
$m\in\M$.  A behavioral sender strategy is denoted
$s_S(m\mid z,\theta,\omega)$.
\item The receiver observes $(z,m)$ but not $(\theta,\omega,c)$, forms a belief
on the complete hidden node, and chooses $a\in\A$ according to
$s_R(a\mid z,m)$.
\item Payoffs are $u_S^\theta(a,\omega)$ and $u_R(a,\omega)$.
\end{enumerate}
At the endpoints, zero-probability source branches are omitted from the
extensive form: the nonbinding sender branch is absent when $\rho=1$, and the
binding branch is absent when $\rho=0$.
Conditional on $\theta$, the calibration sample, fresh deployment state, and
fresh binding shock are independent except through their common dependence on
the fixed structural type.

For a feasible $z$, $\Gamma(z)$ is the conditional Harsanyi game whose active
type set is
\[
\Tset(z)
=\left\{\theta\in\Tset:
\lambda_0(\theta)\prod_{t=1}^N x^\theta_{\omega_t m_t}>0
\right\}.
\]
Equivalently, $\Tset(z)$ is the support of the Bayesian posterior defined in
\eqref{eq:type-posterior}.  Types outside this set are not nodes of the
conditional game.  For each receiver message information set, let
$\mathsf H_z(m)$ denote the set of hidden histories in that information set,
restricted to $\Tset(z)$ and to the source branches present at the relevant
value of $\rho$.

\begin{definition}[Equilibrium object]
\label{def:equilibrium-object}
For a public transcript $z$ of positive prior-predictive probability, a
deployment assessment consists of sender and receiver strategies in
$\Gamma(z)$ and receiver beliefs over $(\theta,\omega,c)$ after each message.
It is a PBE if strategies are sequentially rational at every information set
and beliefs satisfy Bayes' rule at every positive-probability receiver
information set.  After a zero-probability message $m$, beliefs are
probability distributions on $\mathsf H_z(m)$ and the receiver takes a best
response to the specified belief.  We use this weak-PBE convention, with no
sequential-equilibrium restriction on off-path beliefs
~\cite{FudenbergTirole1991}.
\end{definition}

The paper's phrase \emph{post-calibration PBE} always refers to
Definition~\ref{def:equilibrium-object}.  It does not mean that the
institutional production of $Z_N$ is strategic equilibrium play.

\subsection{Notation and uncertainty audit}

The analysis keeps persistent uncertainty, fresh deployment uncertainty,
observable actions, and the three receiver-facing statistical objects
separate.  \Cref{tab:notation} collects their meanings and observability in one
notation audit; in particular, $\widehat\beta_N$ is a statistic,
$\overline\beta_N$ is the receiver's predictive decision belief, and $q_N$ is
the equilibrium belief on the complete hidden node.

\section{Reduced Forms, Identification, and Type-wise Implementability}
\label{sec:reduced-form}
\subsection{Type-wise feasible reduced forms}

Fix a structural type $\theta$.  A joint array $x\in
\mathbb R_+^{\Om\times\M}$ has the correct state marginal if
\begin{equation}
\sum_{m\in\M}x_{\omega m}=\mu(\omega)
\quad\text{for every }\omega.
\label{eq:state-marginal}
\end{equation}
Because $\mu$ has full support, it induces a kernel
$\phi_x(m\mid\omega)=x_{\omega m}/\mu(\omega)$.  Define
\begin{equation}
\calX_\rho^\theta
=\left\{
x\ge0:
\eqref{eq:state-marginal}\text{ holds and }
\sum_{m\notin B_{\omega}^{S,\theta}}x_{\omega m}
\le \rho\mu(\omega)\ \ \forall\omega
\right\}.
\label{eq:type-implementable-polytope}
\end{equation}
The last inequality says that only the binding component can generate
recommendations outside the sender-best set.

\begin{lemma}[Type-wise reduced-form implementability]
\label{lem:typewise-implementability}
Fix $\theta$ and $\rho\in[0,1]$.  A joint array $x$ can be represented as
\[
x_{\omega m}
=\mu(\omega)\left[
\rho\pi(m\mid\omega)+(1-\rho)\sigma(m\mid\omega)
\right],
\qquad
\supp\sigma(\cdot\mid\omega)\subseteq B_{\omega}^{S,\theta},
\]
for stochastic kernels $(\pi,\sigma)$ if and only if
$x\in\calX_\rho^\theta$.
\end{lemma}

The endpoint interpretation is exact.  At $\rho=0$, the reduced form itself
must be supported on sender-best messages.  At $\rho=1$, the binding kernel can
generate any reduced form with state marginal $\mu$, and the discretionary
kernel is payoff irrelevant.  For $0<\rho<1$, the proof constructs a
subdistribution of total mass $1-\rho$ inside $B_{\omega}^{S,\theta}$ and
assigns the remaining mass to the binding kernel.  The construction is given in
\Cref{app:two-stage-proofs}.

Lemma~\ref{lem:typewise-implementability} is a protocol decomposition, not yet
an equilibrium theorem.  Sender optimality additionally requires that a
prescribed sender-best message induce its namesake final action.  That is proved
only after receiver sequential rationality is established in
\Cref{sec:equilibrium-implementation}.

\subsection{Behavioral sufficiency and the identification boundary}

Conditional on a fixed type, hiding the binding shock gives
\[
\Pr(\omega,m\mid\theta)
=x^\theta_{\omega m}.
\]
Consequently, every on-path receiver posterior and every expected payoff under
a fixed receiver response depends on the latent pair only through $x^\theta$.
With type uncertainty, the same statement applies to the posterior-weighted
reduced form.

\begin{proposition}[Reduced-form behavioral sufficiency]
\label{prop:behavioral-sufficiency}
Let $\lambda\in\Delta(\Tset)$ be a belief over structural types and suppose the
installed type-contingent kernels are played.  Define
\[
\overline x^\lambda_{\omega m}
=\sum_{\theta\in\Tset}\lambda(\theta)x^\theta_{\omega m}.
\]
For every $m$ with $p_m(\overline x^\lambda)>0$, the receiver's state
posterior is
$\beta_{\overline x^\lambda}(\omega\mid m)
=\overline x^\lambda_{\omega m}/p_m(\overline x^\lambda)$.
Thus receiver best responses depend on the type posterior and latent kernels
only through $\overline x^\lambda$.
\end{proposition}

Behavioral sufficiency does not by itself determine which structural
coordinates are identified.  Define the reduced-form equivalence class
\begin{equation}
[\theta]=\{\vartheta\in\Tset:x^\vartheta=x^\theta\}.
\label{eq:equivalence-class}
\end{equation}

\begin{proposition}[Population identification in the finite dictionary]
\label{prop:dictionary-identification}
Under certified sampling, the population law identifies $x^\theta$ and the
identified set $[\theta]$.  The structural type is point identified at
$\theta$ if and only if $[\theta]=\{\theta\}$.  More generally, a type
functional $g(\theta)$ is identified if and only if it is constant on
$[\theta]$.  In particular, if $\theta\mapsto x^\theta$ is injective, the
dictionary identifies the stored pair
$(\pi^\theta,\overline\sigma^\theta)$ by lookup.
\end{proposition}

Members of a nonsingleton class may differ in sender payoff, sender-best sets,
and latent decompositions.  No sample generated from $x^\theta$ can separate
those members.  This finite-dictionary statement must be distinguished from
multiplicity in an unrestricted decomposition space.

\begin{proposition}[Nonparametric decomposition multiplicity]
\label{prop:latent-nonidentification}
Suppose $0<\rho<1$.  Fix a state and a reduced-form probability vector
$\phi\in\Delta(\M)$ with $K=\supp\phi$ and $|K|\ge2$.  In the unrestricted
statewise decomposition class with sender-best set $B=K$, the set of stochastic
pairs $(\pi,\overline\sigma)$ satisfying
$\phi=\rho\pi+(1-\rho)\overline\sigma$ and
$\supp\overline\sigma\subseteq B$ contains a continuum.  Holding all other
statewise kernels fixed gives a continuum of full kernel pairs.  This does not
imply that $\Tset$ contains a continuum of candidate types.  If the finite dictionary
does contain two types $\theta_0\ne\theta_1$ with the same $x$, then for any
estimator $\widehat\theta_N$ and any metric $d$ on these types,
\[
\max_{i\in\{0,1\}}
\E_{\theta_i}d(\widehat\theta_N,\theta_i)
\ge \frac12 d(\theta_0,\theta_1).
\]
\end{proposition}

The lower bound follows from observational equivalence and the triangle
inequality, not from slow convergence.  When $[\theta]$ is nonsingleton, it
explains why the receiver's learning target is the class or its predictive
reduced form rather than the full tuple
$(u_S^\theta,B^{S,\theta},\pi^\theta,\overline\sigma^\theta)$.  When the class is
a singleton, the same implementation results continue to apply, but learning
also recovers the dictionary type.

\subsection{Receiver obedience and the type-wise design layer}

For recommendation $m$ and alternative final action $a\ne m$, define the
unnormalized obedience slack
\begin{equation}
D_{m,a}(x)
=\sum_{\omega\in\Om}x_{\omega m}
\left[u_R(m,\omega)-u_R(a,\omega)\right].
\label{eq:obedience-slack}
\end{equation}
For $\gamma\ge0$, the type-wise robust feasible set is
\begin{equation}
\calF_{\rho,\gamma}^{\theta}
=\left\{
x\in\calX_\rho^\theta:
D_{m,a}(x)\ge\gamma p_m(x)
\quad\forall m\in\M,\ a\ne m
\right\}.
\label{eq:type-robust-feasible}
\end{equation}
When $p_m(x)>0$, this is equivalent to a conditional receiver payoff gap of at
least $\gamma$.  Off-support constraints are vacuous.

For a fixed type, define the direct-follow sender value
\[
U_S^\theta(x)
=\sum_{\omega,m}x_{\omega m}u_S^\theta(m,\omega),
\qquad
V_{\rho,\gamma}^\theta
=\max_{x\in\calF_{\rho,\gamma}^{\theta}}U_S^\theta(x)
\]
whenever the feasible set is nonempty.  These are institution or designer
problems conditional on $\theta$.  They select a target protocol subject to
sender participation or installation outside the present game.  Nothing in
the optimization alone proves that a privately informed sender chooses that
protocol at an ex-ante mechanism-selection node.

We impose positive margin only on a target or realized type when proving
finite-sample entry.  The candidate family may include receiver-nonobedient
types.  This asymmetry is essential: if
$x^\theta\in\calF_{\rho,\gamma}^{\theta}$ for every candidate $\theta$ with a
common $\gamma>0$, then every posterior mixture also has the same unnormalized
obedience inequalities, so the receiver should follow before seeing
calibration data.

\section{Certified Bayesian Learning}
\label{sec:certified-learning}
\subsection{Posterior over structural types}

For a transcript $z=((\omega_t,m_t))_{t=1}^N$, let
\[
L_\theta(z)=\prod_{t=1}^N x^\theta_{\omega_t m_t},
\qquad
\ell(z)=\sum_{\vartheta\in\Tset}\lambda_0(\vartheta)L_\vartheta(z).
\]
Call $z$ \emph{feasible} if $\ell(z)>0$, and let
$\calH_N^+=\{z:\ell(z)>0\}$.  At every feasible transcript, Bayes' rule gives
\begin{equation}
\lambda_N(\theta\mid z)
=\frac{\lambda_0(\theta)L_\theta(z)}
{\sum_{\vartheta\in\Tset}\lambda_0(\vartheta)L_\vartheta(z)}.
\label{eq:type-posterior}
\end{equation}
An infeasible transcript has zero probability under every candidate type.  The
gate rejects it and invokes the fallback in \eqref{eq:fallback-action}; no
Bayesian posterior is asserted there.

Equation~\eqref{eq:type-posterior} updates a structural-type posterior, but its
decision content is reduced form.  If the true type is $\theta^\circ$, the
population identified set is $[\theta^\circ]$, and no calibration sample can
distinguish members of that class.  When the class is a singleton, the finite
dictionary therefore does identify $\theta^\circ$; otherwise the posterior can
at most concentrate on the class.  Point identification of the structural type
is not required for implementation.

\subsection{Full-node and posterior-predictive beliefs}

Given a feasible $z$, define the posterior-predictive joint distribution and
message probability
\begin{equation}
\overline x_N(\omega,m\mid z)
=\sum_{\theta\in\Tset}\lambda_N(\theta\mid z)x^\theta_{\omega m},
\qquad
\overline p_N(m\mid z)
=\sum_{\omega}\overline x_N(\omega,m\mid z).
\label{eq:predictive-joint}
\end{equation}
Let
$S_N(z)=\{m:\overline p_N(m\mid z)>0\}$.
If the installed discretionary prescriptions are played, the Bayesian belief
at the receiver's information set $(z,m)$ is, for $m\in S_N(z)$,
\begin{align}
q_N(\theta,\omega,1\mid m,z)
&=
\frac{
\lambda_N(\theta\mid z)\mu(\omega)\rho
\pi^\theta(m\mid\omega)}
{\overline p_N(m\mid z)},                                     
\label{eq:full-belief-binding}\\
q_N(\theta,\omega,0\mid m,z)
&=
\frac{
\lambda_N(\theta\mid z)\mu(\omega)(1-\rho)
\overline\sigma^\theta(m\mid\omega)}
{\overline p_N(m\mid z)}.
\label{eq:full-belief-discretionary}
\end{align}
Summing over $(\theta,c)$ gives the state belief relevant for receiver utility:
\begin{equation}
\overline\beta_N(\omega\mid m,z)
=\sum_{\theta,c}q_N(\theta,\omega,c\mid m,z)
=
\frac{\sum_\theta\lambda_N(\theta\mid z)x^\theta_{\omega m}}
{\sum_\theta\lambda_N(\theta\mid z)p_m(x^\theta)}.
\label{eq:predictive-state-belief}
\end{equation}

\begin{lemma}[Bayesian consistency of the learning bridge]
\label{lem:bayesian-beliefs}
Under Assumptions~\ref{ass:installed}--\ref{ass:certification},
\eqref{eq:type-posterior} follows from the common prior and certified
calibration likelihood.  Conditional on the installed deployment strategies,
\eqref{eq:full-belief-binding}--\eqref{eq:full-belief-discretionary} follow from
the continuation strategies and Bayes' rule at every positive-probability
receiver information set.  Their state marginal is
\eqref{eq:predictive-state-belief}.
\end{lemma}

The full belief $q_N$ and its state marginal serve different proof obligations.
The former establishes PBE consistency; the latter establishes receiver
sequential rationality because $u_R$ depends only on $(a,\omega)$.  Calibration
preserves prior odds among observationally equivalent dictionary types and
need not select a particular member of the class.  At deployment, the message
also need not reveal the latent source $c$.  Neither form of concentration is
necessary for the receiver's action.

\subsection{Empirical estimates are not equilibrium beliefs}

Let $N_N(\omega,m)$ and $N_N(m)$ denote the corresponding counts in $z$.  When
$N_N(m)>0$, define
\begin{equation}
\widehat\beta_N(\omega\mid m)
=\frac{N_N(\omega,m)}{N_N(m)}.
\label{eq:empirical-posterior}
\end{equation}
When $N_N(m)=0$, set $\widehat\beta_N(\cdot\mid m)$ equal to an arbitrary fixed
distribution.  The low-count term in the empirical audit explicitly covers
this event.
This statistic consistently estimates the true type's reduced-form posterior
under standard support conditions.  It is not the belief generated from the
common prior and candidate strategies, and its concentration does not prove PBE
consistency.  The equilibrium receiver uses
$\overline\beta_N$; \Cref{sec:robust-statistical} retains
$\widehat\beta_N$ as an operational audit and finite-sample benchmark.

\subsection{Posterior-predictive certificate and activation algorithm}

For any $\tau\ge0$, define the predictive obedience numerator
\begin{equation}
\overline D_{m,a}(z)
=\sum_{\theta,\omega}\lambda_N(\theta\mid z)x^\theta_{\omega m}
\left[u_R(m,\omega)-u_R(a,\omega)\right]
\label{eq:predictive-obedience}
\end{equation}
and the public acceptance set
\begin{equation}
\calC_N(\tau)
=\left\{
z\in\calH_N^+:
\overline D_{m,a}(z)\ge
\tau\,\overline p_N(m\mid z)
\quad
\forall m\in S_N(z),\ a\ne m
\right\}.
\label{eq:acceptance-set}
\end{equation}
The quantifier covers \emph{every} positive posterior-predictive message.  It
cannot be replaced by the empirical support or by messages above a numerical
probability threshold when the goal is exact PBE: any positive-probability
message is on path.

\begin{algorithm}[H]
\caption{Certified calibration and PBE activation}
\label{alg:activation}
\small
\textbf{Input.} The common type dictionary, prior $\lambda_0$, transcript $z$,
and target margin $\tau\ge0$.
\begin{enumerate}[leftmargin=1.8em,label=\arabic*.,itemsep=0.15em,topsep=0.3em]
\item \textbf{Feasibility.} Compute the likelihoods $L_\theta(z)$.  If
$\ell(z)=0$, reject.
\item \textbf{Posterior.} Compute $\lambda_N(\cdot\mid z)$ in
\eqref{eq:type-posterior} and $\overline x_N,\overline p_N$ in
\eqref{eq:predictive-joint}.
\item \textbf{Obedience test.} For every $m$ with
$\overline p_N(m\mid z)>0$ and every $a\ne m$, compute
$\overline D_{m,a}(z)$.
\item \textbf{Activation.} Activate the installed deployment channel if and
only if $z\in\calC_N(\tau)$; otherwise use the no-communication fallback.
\item \textbf{Deployment.} After activation, the binding device uses
$\pi^\theta$ when $c=1$; when $c=0$, use the equilibrium sender and receiver
strategies proved in \Cref{thm:post-calibration-pbe}.
\end{enumerate}
\end{algorithm}

With explicitly tabulated finite sets, direct evaluation uses
\[
O\!\left(
N|\Tset|+|\Tset||\Om||\M|+|\Om||\M||\A|
\right)
\]
arithmetic operations: the three terms compute type likelihoods, the predictive
joint law, and all obedience gaps.  Exact rational comparisons have polynomial
bit complexity in the encoded input and $N$.  Numerical implementations should
accumulate log likelihoods and normalize with a log-sum-exp calculation to
avoid underflow; this changes neither the posterior nor the gate.

The certificate checks receiver incentives rather than asserting them.  Its
substantive force comes from the entry theorem in
\Cref{thm:high-probability-entry}, which gives conditions under which
$\calC_N(\tau)$ is reached with high probability.

\section{Static and Post-calibration Equilibrium Implementation}
\label{sec:equilibrium-implementation}
\subsection{A static posterior-mixture implementation theorem}

The type-wise decomposition in
Lemma~\ref{lem:typewise-implementability} does not by itself establish
equilibrium under type uncertainty.  The receiver has one action rule and must
best respond to the posterior mixture of all active types.  The following
static theorem isolates the exact condition needed at a deployment
continuation.

\begin{theorem}[Bayesian direct-following implementation]
\label{thm:static-mixture-implementation}
Fix a belief $\lambda$ on a nonempty active type set
$\Tset_\lambda=\{\theta:\lambda(\theta)>0\}$.  Suppose
\eqref{eq:sender-best-prescription} holds for every active type and define
$\overline x^\lambda=\sum_\theta\lambda(\theta)x^\theta$.  If
\begin{equation}
D_{m,a}(\overline x^\lambda)\ge0
\quad\text{for every }m\text{ with }
p_m(\overline x^\lambda)>0\text{ and every }a\ne m,
\label{eq:mixture-obedience}
\end{equation}
then the one-shot hidden-binding game admits a PBE in which:
\begin{enumerate}[leftmargin=1.8em,label=(\roman*)]
\item at every nonbinding sender node, the sender uses
$\overline\sigma^\theta(\cdot\mid\omega)$;
\item the receiver chooses action $m$ after every on-path recommendation $m$;
\item on-path beliefs are the Bayes beliefs induced by $\lambda$,
$\pi^\theta$, and $\overline\sigma^\theta$;
\item after an off-path message, the receiver holds a specified probability
distribution on the histories in that information set and takes a best
response to that belief.
\end{enumerate}
Conditional on each active $\theta$, the equilibrium induces $x^\theta$.
If every inequality in \eqref{eq:mixture-obedience} is strict, direct following
is the unique receiver best response on path.
\end{theorem}

The complete proof, including the off-path completion and sender incentives in
the nonbinding branch, is in \Cref{app:static-mixture-proof}.

\begin{corollary}[Known-type direct-following implementation]
\label{cor:known-type-implementation}
Fix $\theta$.  If
$x^\theta\in\calF_{\rho,0}^\theta$ and the decomposition has
sender-best discretionary support, then a direct-following PBE implements
$x^\theta$.  If $x^\theta\in\calF_{\rho,\gamma}^\theta$ for $\gamma>0$, the
on-path receiver response is unique.
\end{corollary}

The corollary is the static implementation foundation for the type-wise design
polytope.  It is not sufficient for the uncertain-type continuation; the latter
uses Theorem~\ref{thm:static-mixture-implementation} and predictive mixture
obedience.

\subsection{Exact PBE after every accepted transcript}

\begin{theorem}[Exact post-calibration PBE]
\label{thm:post-calibration-pbe}
Fix $\tau\ge0$ and a feasible transcript $z\in\calC_N(\tau)$.  In
$\Gamma(z)$, let the nonbinding sender use
\[
s_S(m\mid z,\theta,\omega)
=\overline\sigma^\theta(m\mid\omega)
\quad\text{for every }\theta\in\supp\lambda_N(\cdot\mid z),
\]
and let the receiver choose $m$ after every $m\in S_N(z)$.  Use
$q_N$ from
\eqref{eq:full-belief-binding}--\eqref{eq:full-belief-discretionary} on path;
after every $m\notin S_N(z)$, assign an arbitrary belief supported on
$\mathsf H_z(m)$ and a receiver best response.  This assessment is an exact
PBE of the deployment continuation.
Conditional on the realized active type $\theta$, it implements
$x^\theta$.  If $\tau>0$, direct following is uniquely optimal on path with
receiver gap at least $\tau$.
\end{theorem}

The proof is given in \Cref{app:post-calibration-proof}.

Posterior-zero types are not active in $\Gamma(z)$ and are not used to assert
conditional implementation.  The rejected branch is the no-message decision
problem in \eqref{eq:fallback-action}.  Thus all public transcripts have a
specified continuation, while the target implementation claim is made only on
feasible accepted transcripts.

\subsection{When does calibration reach the PBE region?}

The deterministic implication in
Theorem~\ref{thm:post-calibration-pbe} is not enough: the acceptance set could
be empty.  Two additional conditions make entry nonvacuous.

\begin{assumption}[Common recommendation support]
\label{ass:common-message-support}
There is a nonempty $S\subseteq\M$ such that, for every
$\theta\in\supp\lambda_0$,
\[
p_m(x^\theta)>0\quad\Longleftrightarrow\quad m\in S.
\]
\end{assumption}

This assumption does not make learning redundant.  Types may use the same
message labels while inducing different, including nonobedient, state
posteriors.  It rules out a subtler finite-sample problem: a false type with
positive posterior can otherwise introduce a message absent under the target
type, and target-type margin says nothing about following that message.

Fix a true type $\theta^\circ$ and write
$P=x^{\theta^\circ}$, $E=[\theta^\circ]$, and
$w_E=\lambda_0(E)$.  If $E\ne\Tset$, define the maximal Bhattacharyya
affinity, its logarithmic separation, and the prior factor
\begin{align}
b^\circ
&=\max_{\vartheta\notin E}
\sum_{(\omega,m)}
\sqrt{P_{\omega m}x^\vartheta_{\omega m}}\in[0,1),
\label{eq:bhattacharyya-affinity}\\
h^\circ&=-\log b^\circ\in(0,\infty],
\label{eq:hellinger-separation}\\
A_E
&=\sum_{\vartheta\notin E}
\sqrt{\frac{\lambda_0(\vartheta)}{w_E}}.
\label{eq:prior-factor}
\end{align}
Distinct finite distributions have affinity below one.  Zero cells and mutual
singularity are allowed.

\begin{lemma}[Finite-sample class concentration]
\label{lem:class-concentration}
Let
$r_N=\lambda_N(E^c\mid Z_N)$.  Conditional on true type
$\theta^\circ$, if $E\ne\Tset$, then
\begin{equation}
\E_{\theta^\circ}[r_N]\le A_E(b^\circ)^N,
\qquad
\Pr_{\theta^\circ}\left(
r_N>\frac{A_E(b^\circ)^N}{\eta}
\right)\le\eta
\label{eq:class-concentration}
\end{equation}
for every $\eta\in(0,1)$, where $(b^\circ)^0=1$.  If $E=\Tset$, then
$r_N=0$ identically.
\end{lemma}

\begin{lemma}[Predictive convergence and margin preservation]
\label{lem:predictive-margin}
Suppose Assumption~\ref{ass:common-message-support} holds.  Let
\[
p_{\min}^\circ=\min_{m\in S}p_m(P)>0,\qquad
R_R=\max_{\omega,a,b}|u_R(a,\omega)-u_R(b,\omega)|.
\]
For every $m\in S$,
\begin{equation}
\left\|
\overline\beta_N(\cdot\mid m,Z_N)-\beta_P(\cdot\mid m)
\right\|_1
\le
\frac{2r_N}{(1-r_N)p_{\min}^\circ}.
\label{eq:predictive-beta-bound}
\end{equation}
If $P\in\calF_{\rho,\gamma}^{\theta^\circ}$ for $\gamma>0$ and
\begin{equation}
r_N\le
\alpha^\circ
:=\frac{\gamma p_{\min}^\circ}
{4R_R+\gamma p_{\min}^\circ},
\label{eq:posterior-threshold}
\end{equation}
then $Z_N\in\calC_N(\gamma/2)$.
\end{lemma}

\begin{theorem}[High-probability entry into an exact PBE]
\label{thm:high-probability-entry}
Suppose Assumptions~\ref{ass:installed}--\ref{ass:common-message-support}
hold and the true reduced form
$P=x^{\theta^\circ}$ lies in
$\calF_{\rho,\gamma}^{\theta^\circ}$ for some $\gamma>0$.  If
$E\ne\Tset$, $b^\circ\in(0,1)$, and
\begin{equation}
N\ge
\frac{1}{h^\circ}
\log\left(\frac{A_E}{\eta\alpha^\circ}\right),
\label{eq:bayesian-sample-complexity}
\end{equation}
with a nonpositive logarithm interpreted as requiring $N\ge0$, then
\begin{equation}
\Pr_{\theta^\circ}
\left(Z_N\in\calC_N(\gamma/2)\right)
\ge1-\eta.
\label{eq:high-probability-entry}
\end{equation}
If $E=\Tset$, the conclusion holds for every feasible $N$, including $N=0$.
If $E\ne\Tset$ and $b^\circ=0$, the same conclusion holds for every
$N\ge1$ because one observation eliminates every outside type almost surely;
no claim for $N=0$ is implied by infinite separation.
On the event in \eqref{eq:high-probability-entry}, the activated continuation
is the exact PBE in Theorem~\ref{thm:post-calibration-pbe}.
\end{theorem}

The proof is given in \Cref{app:high-probability-proof}.

\paragraph{Why common support appears.}
Consider one state, messages/actions $\{0,1\}$, and a receiver who strictly
prefers $0$.  Let $\rho=1/2$.  A true type always emits $0$; a false type has a
binding kernel that emits $0$ and a sender-best discretionary kernel that emits
$1$.  After $N$ observed zeros, the false type's posterior is positive but
converges to zero.  At every finite $N$, message $1$ nevertheless has positive
predictive probability and following it is strictly suboptimal.  Hence
true-type margin and posterior consistency alone do not imply finite-sample
entry into an exact direct-following PBE.  Common support is a transparent
sufficient repair; exclusive-message safety or finite-time type elimination
would be alternatives.

\subsection{Statistical failure versus strategic approximation}

The parameter $\eta$ in Theorem~\ref{thm:high-probability-entry} is the
probability of not reaching the certified region.  It is not a deviation-gain
bound, so the theorem does not define an ``$\eta$-PBE.''

For completeness, call a continuation an $\varepsilon$-PBE if no type-agent or
receiver information set admits a deviation improving conditional payoff by
more than $\varepsilon$.  Suppose sender prescriptions are
$\varepsilon_S$-best statewise and the gate accepts only when every predictive
receiver gap is at least $-\varepsilon_R$.  The proof of
Theorem~\ref{thm:static-mixture-implementation} then gives an
$\varepsilon$-PBE with
$\varepsilon=\max\{\varepsilon_S,\varepsilon_R\}$.  If such a continuation is
reached with probability at least $1-\eta$, we call the pair
$(\varepsilon,\eta)$ a \emph{certified implementation guarantee}.  In the ideal
model of this paper, $\varepsilon=0$; $\eta$ remains a separate entry-failure
probability.

\section{Robust Value and Statistical Certification}
\label{sec:robust-statistical}
\subsection{The type-wise value--margin frontier}

Fix a type $\theta$.  Let
\[
\calF_\rho^{\theta,+}=\bigcup_{\gamma>0}\calF_{\rho,\gamma}^\theta,\qquad
V_\rho^{\theta,+}=\sup_{x\in\calF_\rho^{\theta,+}}U_S^\theta(x),\qquad
V_\rho^\theta=V_{\rho,0}^\theta.
\]
The fixed-margin sets are nested, so $V_{\rho,\gamma}^\theta$ is weakly
decreasing in $\gamma$ on the values for which the feasible set is nonempty.
The positive-margin benchmark need not be attained because
$\calF_\rho^{\theta,+}$ is an open union relative to message supports.

\begin{theorem}[Positive-margin frontier]
\label{thm:positive-margin-frontier}
If $\calF_\rho^{\theta,+}\ne\varnothing$, then
\[
\lim_{\gamma\to0^+}V_{\rho,\gamma}^\theta=V_\rho^{\theta,+}.
\]
Moreover, for every $\gamma>0$ such that
$\calF_{\rho,\gamma}^\theta\ne\varnothing$,
\begin{equation}
V_\rho^\theta-V_{\rho,\gamma}^\theta
=
\left(V_\rho^\theta-V_\rho^{\theta,+}\right)
{}+\left(V_\rho^{\theta,+}-V_{\rho,\gamma}^\theta\right).
\label{eq:frontier-decomposition}
\end{equation}
\end{theorem}

The first term is a strictification discontinuity; the second is the local
cost of a quantitative margin.  Both affect certification.  A zero
strictification cost does not imply that the supremum is attained by a
positive-margin scheme.

\begin{assumption}[Type-wise support-preserving strictifiability]
\label{ass:strictifiability}
For the fixed type, there exist a weak optimum $x^\star$, a strict point $y$,
a common message support $H$, and $\xi>0$ such that
\[
x^\star\in\argmax_{x\in\calF_{\rho,0}^\theta}U_S^\theta(x),\qquad
y\in\calX_\rho^\theta,\qquad
p_m(y)\ge\xi,\quad D_{m,a}(y)\ge\xi
\]
for every $m\in H$ and $a\ne m$, while both $x^\star$ and $y$ assign zero
message probability outside $H$.
\end{assumption}

\begin{theorem}[Local linear robustification]
\label{thm:linear-robustification}
Under Assumption~\ref{ass:strictifiability}, for
$0<\gamma\le\xi$,
\[
0\le V_\rho^\theta-V_{\rho,\gamma}^\theta
\le\frac{\gamma}{\xi}
\]
when $u_S^\theta\in[0,1]$.  This order is worst-case tight: for every
$\rho\in(0,1]$ there is a two-state, two-action instance with
\[
V_\rho^\theta=V_\rho^{\theta,+}=1,\qquad
V_{\rho,\gamma}^\theta=\frac{1}{1+\gamma}
\]
for $0<\gamma\le\rho/(2-\rho)$.
\end{theorem}

The proof mixes a weak optimum with a support-preserving strict anchor using
weight $\gamma/\xi$; the matching lower-order example is in
\Cref{app:frontier-algorithm-proofs}.  These results are unchanged
mathematically from the reduced-form paper, but their interpretation changes:
the margin protects receiver optimality and, holding message frequency and
type separation fixed, shortens the stated calibration bound.

\subsection{A nonparametric empirical audit}

The Bayesian entry rate in
\eqref{eq:bayesian-sample-complexity} is parametric: it depends on separation of
the finite candidate family.  The original empirical estimator gives a
different, model-agnostic audit rate.  Let
$d=|\Om|$, $n=|\M|$, $C_d=\max\{1,2^d-2\}$, and suppose the true reduced form
$P$ has on-path
support $S$ and
$p_{\min}^\circ=\min_{m\in S}p_m(P)>0$.

\begin{proposition}[Uniform empirical posterior audit]
\label{prop:empirical-audit}
For every $\varepsilon>0$,
\begin{align}
\Pr_{\theta^\circ}\left[
\max_{m\in S}
\|\widehat\beta_N(\cdot\mid m)-\beta_P(\cdot\mid m)\|_1
>\varepsilon
\right]
\le{}&
n\exp\left(-\frac{Np_{\min}^\circ}{8}\right)
\nonumber\\
&{}+nC_d
\exp\left(-\frac{Np_{\min}^\circ\varepsilon^2}{4}\right).
\label{eq:empirical-audit-bound}
\end{align}
Consequently, if $P$ has receiver margin $\gamma$ and
$\varepsilon=\gamma/(2R_R)$, it suffices for failure probability at most
$\eta$ that
\begin{equation}
N\ge
\max\left\{
\frac{8}{p_{\min}^\circ}\log\frac{2n}{\eta},
\frac{16R_R^2}{p_{\min}^\circ\gamma^2}
\log\frac{2nC_d}{\eta}
\right\}.
\label{eq:empirical-sample-complexity}
\end{equation}
\end{proposition}

The two terms in \eqref{eq:empirical-audit-bound} respectively ensure that each
message is observed often enough and apply finite-alphabet $L_1$
concentration conditional on that count~\cite{Weissman2003}.  The resulting
$\widetilde O(R_R^2/(p_{\min}^\circ\gamma^2))$ rate is not the same object as
the Hellinger-separated Bayesian rate.  In particular,
$\widehat\beta_N$ does not replace $q_N$ in the PBE proof.  A conservative
empirical implementation may require both the empirical audit and the exact
predictive test \eqref{eq:acceptance-set}; only the latter certifies every
accepted history.

\subsection{Virtual prefix errors and the retained pseudo-regret bound}

To preserve the original online statistic without turning Stage I into
strategic play, imagine the following payoff-neutral diagnostic.  At date $t$,
the receiver first observes the current recommendation $m_t$, computes a
\emph{virtual} greedy action using only the preceding record $Z_{t-1}$, and only
then has $\omega_t$ revealed and added to the record.  The action is not taken
and has no payoff.  Let $M_N(P)$ be the number of dates at which that virtual
action differs from the current recommendation, for $N\ge1$.  For each on-path
$m$, let
\[
\gamma_m(P)
=\min_{a\ne m}
\sum_\omega\beta_P(\omega\mid m)
[u_R(m,\omega)-u_R(a,\omega)].
\]

\begin{proposition}[Virtual stabilization and conditional inference loss]
\label{prop:virtual-regret}
Suppose $n\ge2$ and $\gamma_m(P)>0$ for every $m\in S$.  There is a universal
constant $C$ such that
\begin{equation}
\E[M_N(P)]
\le
\sum_{m\in S}
\min\left\{
Np_m(P),
\frac{CR_R^2}{\gamma_m(P)^2}
\left[d+\log(nN)\right]
\right\}+1.
\label{eq:virtual-mistake-bound}
\end{equation}
If $P\in\calF_{\rho,\gamma}^{\theta^\circ}$, this is at most
$CnR_R^2[d+\log(nN)]/\gamma^2+1$.  The receiver's conditional pseudo-regret and
the sender's realized positive virtual shortfall, both measured relative to
direct following with payoffs in $[0,1]$, are bounded in expectation by the same
expected error count.  Neither is an actual Stage-I payoff because the
diagnostic action is never taken.
\end{proposition}

The result is a time-uniform version of the empirical concentration argument:
after a message has occurred on the order of
$R_R^2[d+\log(nN)]/\gamma_m(P)^2$ times, every later virtual decision is
correct on a high-probability event.  We call
\eqref{eq:virtual-mistake-bound} a stabilization statistic.  It is neither
Stage-I utility regret nor convergence of a PBE.  In a separate repeated
uncertified deployment model, the same expression can bound realized learning
loss, but then sender incentives during data generation must be modeled anew.

\subsection{What a sample-aware designer must trade off}

The margin $\gamma$ is only one determinant of activation.  Holding message
frequency and cross-type separation fixed, a larger margin shortens the stated
calibration bounds.  In a design problem those quantities may change together.
The nonparametric
audit also depends on $p_{\min}(x)$, while Bayesian type learning depends on
cross-type separation.  For a target family, a more faithful batch objective
has the schematic form
\[
\text{design loss}
{}+\Pr(\text{rejection or certification failure})
\times\text{payoff range}.
\]
Support-wise lower bounds $p_m(x)\ge\underline p$ remain linear once a support
is fixed.  Jointly choosing all type-contingent reduced forms to preserve a
common support and sufficient cross-type separation is a family-design problem,
however, and is not solved by optimizing each
$V_{\rho,\gamma}^\theta$ independently.

\section{Support-wise Computation and Binary Actions}
\label{sec:computation}
\subsection{Fixed-margin linear programming}

For a fixed structural type $\theta$ and fixed $(\rho,\gamma)$, feasibility of
$\calF_{\rho,\gamma}^\theta$ and, whenever it is nonempty, the value
$V_{\rho,\gamma}^\theta$ are obtained from a linear program in the
$|\Om||\M|$ entries of $x$.  The constraints are the state marginals,
nonnegativity, type-wise implementability
\eqref{eq:type-implementable-polytope}, and the linear obedience inequalities
\[
D_{m,a}(x)-\gamma p_m(x)\ge0.
\]
Thus a fixed-margin target can be computed in time polynomial in the explicit
finite input and encoding length.  The more delicate problem is diagnosing the
open positive-margin frontier and whether its supremum is attained.

\subsection{Support-wise frontier diagnosis}

Fix $\theta$ and let $H\subseteq\M$ be nonempty.  Define
\begin{align*}
K_H^\theta
=\{x\in\calX_\rho^\theta:\;&
p_m(x)=0\ \forall m\notin H,\\
&D_{m,a}(x)\ge0\ \forall m\in H,\ a\ne m\},
\\
K_H^{\theta,+}
=\{x\in K_H^\theta:\;&
p_m(x)>0,\ D_{m,a}(x)>0\
\forall m\in H,\ a\ne m\}.
\end{align*}
The positive-margin set is the union of
$K_H^{\theta,+}$ over nonempty supports.

For each $H$, use three LP oracles.
\begin{enumerate}[leftmargin=1.8em,label=\textup{LP\arabic*:}]
\item Maximize a common strictness variable $\xi$ over $K_H^\theta$ subject
to $p_m(x)\ge\xi$ and $D_{m,a}(x)\ge\xi$ for all relevant pairs.  Denote the
optimum by $\xi_{H,\theta}^{\mathrm{feas}}$.
\item Compute the closure value
$W_H^\theta=\max_{x\in K_H^\theta}U_S^\theta(x)$.
\item Repeat LP1 on the LP2 optimal face
$\{x\in K_H^\theta:U_S^\theta(x)=W_H^\theta\}$; denote the optimum by
$\xi_{H,\theta}^{\mathrm{opt}}$.
\end{enumerate}

\begin{theorem}[Support-wise positive-margin characterization]
\label{thm:supportwise-characterization}
If $\calF_\rho^{\theta,+}\ne\varnothing$, then
\[
V_\rho^{\theta,+}
=\max_{H:\xi_{H,\theta}^{\mathrm{feas}}>0}W_H^\theta.
\]
Moreover, the positive-margin supremum is attained if and only if some support
$H$ satisfies
\[
\xi_{H,\theta}^{\mathrm{feas}}>0,\qquad
W_H^\theta=V_\rho^{\theta,+},\qquad
\xi_{H,\theta}^{\mathrm{opt}}>0.
\]
\end{theorem}

The proof mixes any weak point on a support with the strict point returned by
LP1, showing that $K_H^{\theta,+}$ is dense in $K_H^\theta$ whenever LP1 is
strictly feasible.  LP3 then separates attainment from a value attained only
on the zero-margin boundary.

\begin{algorithm}[H]
\caption{Type-wise frontier selection}
\label{alg:frontier}
\small
\textbf{Input.} A fixed structural type $\theta$.
\begin{enumerate}[leftmargin=1.8em,label=\arabic*.,itemsep=0.15em,topsep=0.3em]
\item \textbf{Weak benchmark.} Solve the weak LP for $V_{\rho,0}^\theta$.
\item \textbf{Support screening.} Enumerate nonempty $H\subseteq\M$ and
discard supports for which LP1 has
$\xi_{H,\theta}^{\mathrm{feas}}\le0$.
\item \textbf{Strict optimization.} Run LP2 and LP3 on the remaining supports.
\item \textbf{Frontier output.} Return $V_\rho^{\theta,+}$, an attaining
strict optimizer when one exists, or a nonattainment certificate.
\item \textbf{Target margin.} For a desired deployment margin, solve the
direct fixed-$\gamma$ LP, report infeasibility if its feasible set is empty,
and otherwise decompose its output with
Lemma~\ref{lem:typewise-implementability}.
\end{enumerate}
\end{algorithm}

Each oracle is polynomial for fixed $H$, but support enumeration may inspect
$2^{|\M|}-1$ supports.  Algorithm~\ref{alg:frontier} is therefore an exact
finite characterization, not a polynomial-time result in the number of
actions.  When the high-probability PBE theorem is imposed on a family of
types, the independently computed targets must additionally be checked for the
common-support requirement in
Assumption~\ref{ass:common-message-support}.  The type-wise oracle does not
solve that joint family-design problem.

\subsection{Binary actions: a bounded fractional-knapsack oracle}

The fixed-margin binary case admits a sharper specialization under additional
monotonicity assumptions.  Fix $\theta$, let $\A=\M=\{0,1\}$, and suppose
\[
s_\omega^\theta
:=u_S^\theta(1,\omega)-u_S^\theta(0,\omega)>0
\quad\forall\omega.
\]
Thus $B_{\omega}^{S,\theta}=\{1\}$.  Let
\[
z_\omega=\phi^\theta(1\mid\omega),\qquad
d_\omega=u_R(1,\omega)-u_R(0,\omega),\qquad
\underline z_\rho=1-\rho.
\]
Type-wise implementability is the box
$\underline z_\rho\le z_\omega\le1$.  Assume also that the receiver robustly prefers action
$0$ at the prior:
\begin{equation}
\sum_\omega\mu(\omega)(-d_\omega-\gamma)\ge0.
\label{eq:prior-zero-robust}
\end{equation}
Under \eqref{eq:prior-zero-robust}, obedience of recommendation $1$ implies
obedience of recommendation $0$, so only the former constrains the optimum.

Partition the states as
\[
\Om_\gamma^+=\{\omega:d_\omega\ge\gamma\},\qquad
\Om_\gamma^-=\{\omega:d_\omega<\gamma\}
\]
and define
\begin{equation}
C_\gamma^\rho
=
\sum_{\omega\in\Om_\gamma^+}\mu(\omega)(d_\omega-\gamma)
-\underline z_\rho\sum_{\omega\in\Om_\gamma^-}\mu(\omega)(\gamma-d_\omega).
\label{eq:knapsack-capacity}
\end{equation}

\begin{proposition}[Type-wise binary knapsack oracle]
\label{prop:binary-knapsack}
Under the assumptions above, the fixed-$(\rho,\gamma)$ problem is
infeasible if $C_\gamma^\rho<0$.  If $C_\gamma^\rho\ge0$, every optimum sets
$z_\omega=1$ on $\Om_\gamma^+$ and writes
$z_\omega=\underline z_\rho+y_\omega$ on $\Om_\gamma^-$, where $y$ solves
\begin{equation}
\begin{array}{ll}
\displaystyle\max_{0\le y_\omega\le\rho}
&
\displaystyle\sum_{\omega\in\Om_\gamma^-}
\mu(\omega)s_\omega^\theta y_\omega
\\[0.8em]
\mathrm{subject\ to}
&
\displaystyle\sum_{\omega\in\Om_\gamma^-}
\mu(\omega)(\gamma-d_\omega)y_\omega
\le C_\gamma^\rho.
\end{array}
\label{eq:binary-knapsack}
\end{equation}
Sorting bad states in decreasing order of
$s_\omega^\theta/(\gamma-d_\omega)$ and filling to the cap $\rho$ gives an
optimum with at most one fractional state.  The running time is
$O(|\Om|\log|\Om|)$.
\end{proposition}

The reduction is not a claim about every binary persuasion instance.  It uses
a fixed type, statewise strict sender preference for action $1$, the
prior-robust condition \eqref{eq:prior-zero-robust}, and a fixed margin.
State-dependent sender-best actions or sender payoff ties require the general
LP.  The oracle produces an installed design target; the equilibrium guarantee
still comes from \Cref{thm:post-calibration-pbe}.

\section{Computational Illustrations}
\label{sec:computational-evidence}
\subsection{Matched type-wise design instances}

We regenerate the inherited design diagnostics because the archived figures
were not accompanied by code, raw draws, a seed, or a solver tolerance.  The
new experiment fixes $\gamma=0.05$, sets $|\Om|=|\A|=n$ and $\mu$ uniform, and
draws receiver payoffs independently from $[0,1]$.  Sender payoff is
$\alpha u_R+(1-\alpha)\varepsilon$, where $\varepsilon$ is an independent
$[0,1]$ draw, $\alpha$ ranges over seven equally spaced alignment values, and
there are 16 replications per value.  This gives 112 types for each
$n\in\{2,\ldots,10\}$.  Crucially, the same type is solved at every
$\rho\in\{0,0.1,\ldots,0.9\}$ using the fixed-margin LP.  The experiment therefore
contains 10,080 matched LP evaluations, rather than 10,080 independent types.
For inference in the $n=4$ diagnostic, we treat the seven alignment levels as
fixed strata and target their equal-weight mixture.  Thus feasibility at each
$\rho$ estimates the average of the seven stratum-specific feasibility
probabilities.  Conditional sender value is the corresponding ratio of the
equal-weight feasible-value mixture to the equal-weight feasibility
probability, rather than an unqualified i.i.d. mean.

\begin{figure}[!htbp]
\centering
\begin{subfigure}{0.48\textwidth}
\centering
\includegraphics[width=\linewidth]{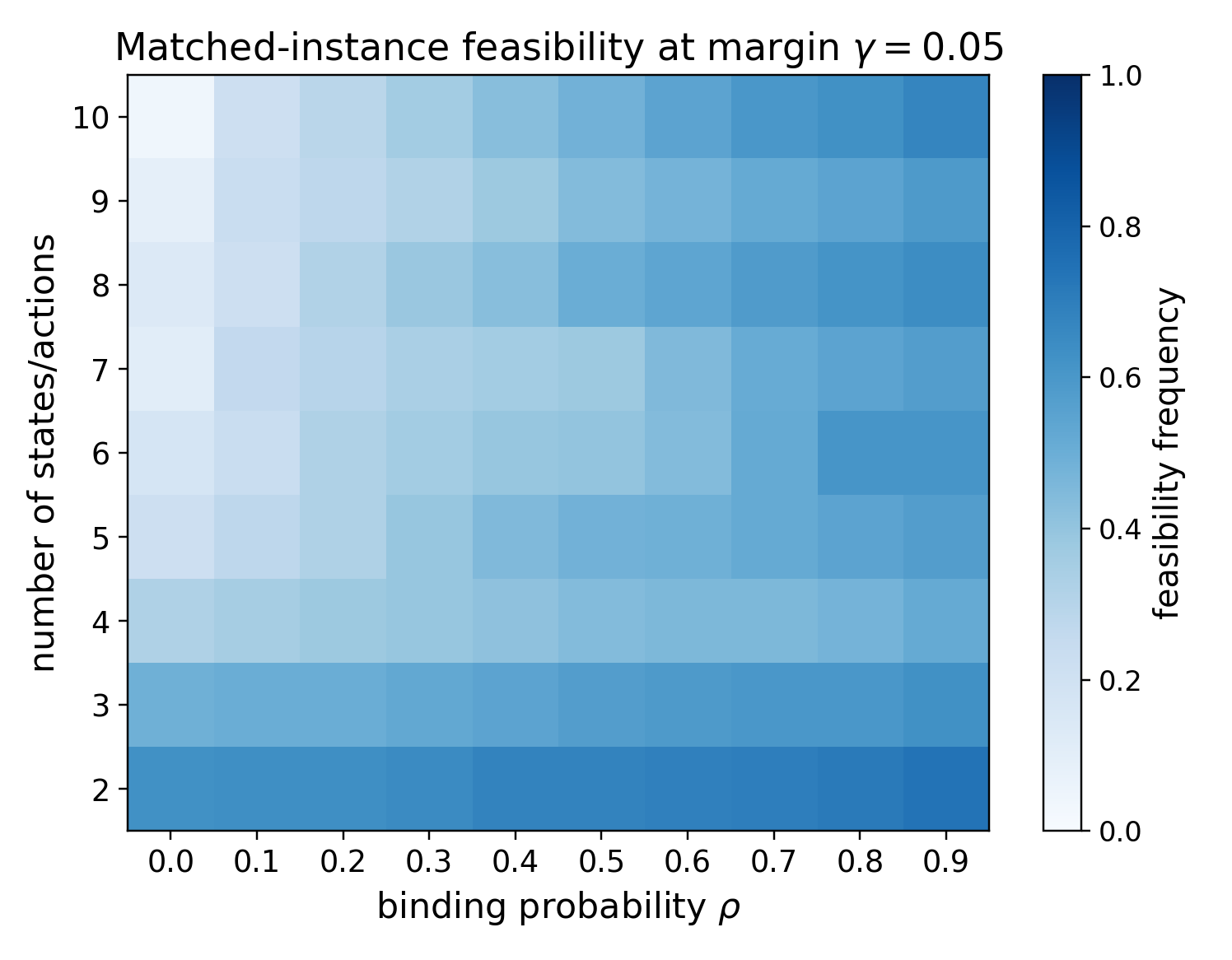}
\caption{Matched fixed-margin feasibility}
\end{subfigure}
\hfill
\begin{subfigure}{0.48\textwidth}
\centering
\includegraphics[width=\linewidth]{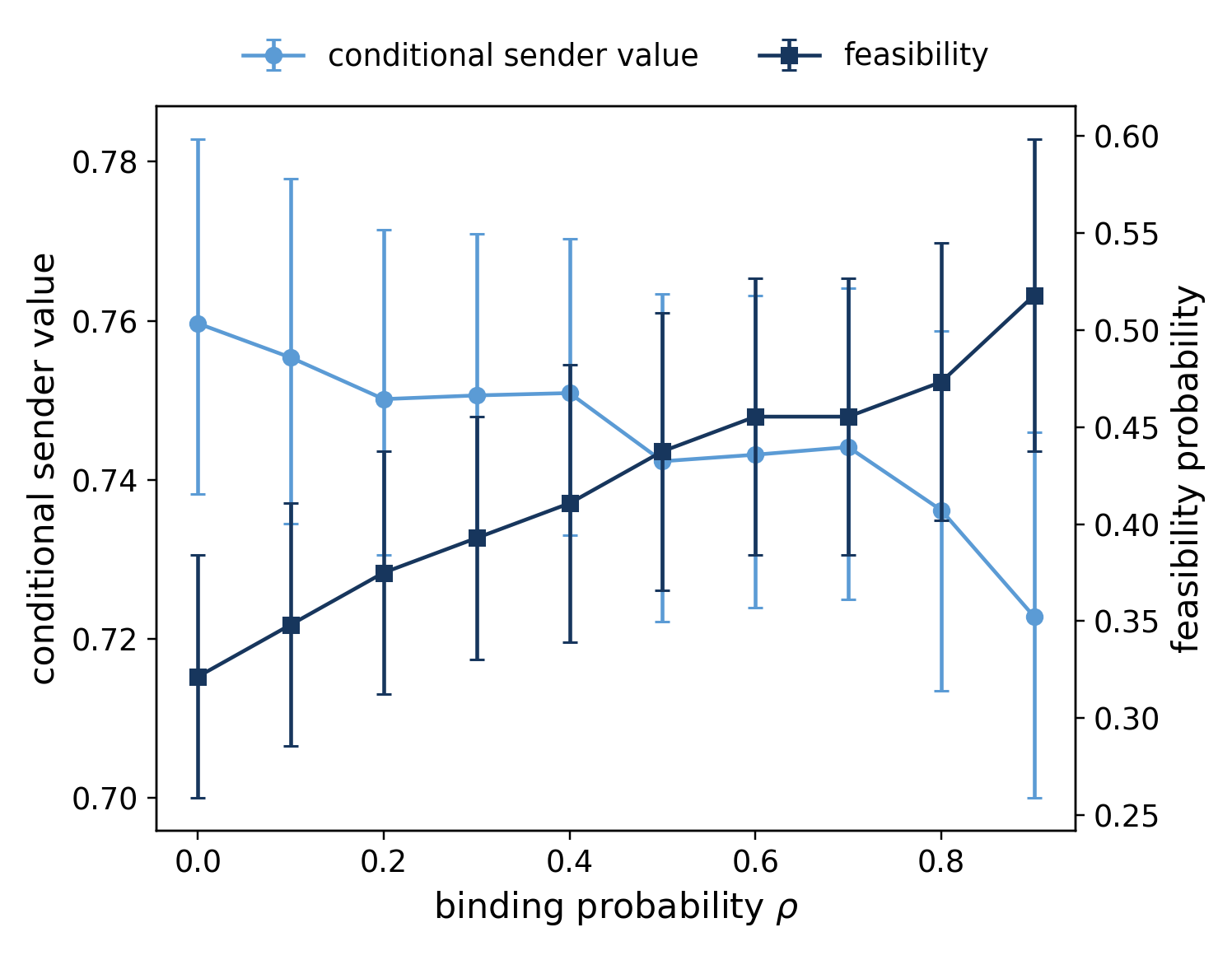}
\caption{The $n=4$ feasibility/value diagnostic}
\end{subfigure}
\caption{Type-wise design diagnostics at $\gamma=0.05$.  Panel (a) reports
feasibility frequency across the 112 matched types.  Since
$\calX_{\rho_1}^\theta\subseteq\calX_{\rho_2}^\theta$ for
$\rho_1\le\rho_2$, fixed-type feasibility and value are weakly increasing in
$\rho$; the code checks both implications.  Panel (b) reports pointwise 95\%
stratified percentile-bootstrap intervals from 20,000 draws.  Each draw
resamples the 16 payoff instances independently within each fixed alignment
stratum and carries an instance's entire $\rho$ path; the conditional-value
ratio, including its random denominator, is recomputed in the same draw.  The
conditional mean can fall when newly feasible, lower-value types enter the
conditioning cohort; this is a composition effect, not a failure of fixed-type
monotonicity.}
\label{fig:matched-design}
\end{figure}

These calculations illustrate the geometry and numerical behavior of the
type-wise design oracle.  They do not simulate calibration, sender learning, or
equilibrium play.  They also do not solve the joint cross-type design problem
needed to impose common support and statistical separation simultaneously.

\subsection{Payoff-neutral prefix stabilization}

For continuity with the empirical learning result, let
$\Om=\A=\{1,\ldots,n\}$, $\mu$ be uniform, and
$u_S(a,\omega)=u_R(a,\omega)=\1\{a=\omega\}$.  The controlled reduced form is
\[
\phi(m\mid\omega)
=
\begin{cases}
\gamma+(1-\gamma)/n,&m=\omega,\\
(1-\gamma)/n,&m\ne\omega.
\end{cases}
\]
It has $p_m=1/n$ and conditional receiver gap $\gamma$.  At
$\rho=1-\gamma$, it is implemented by a deterministic sender-best
discretionary kernel and a uniform binding kernel.  We use $\gamma=0.7$,
$n\in\{2,4,6,8\}$, 2,000 independent calibration paths, and the chronological
virtual rule defined before Proposition~\ref{prop:virtual-regret}.  No action is
taken and no payoff is earned during calibration.

\begin{figure}[!htbp]
\centering
\includegraphics[width=0.90\textwidth]{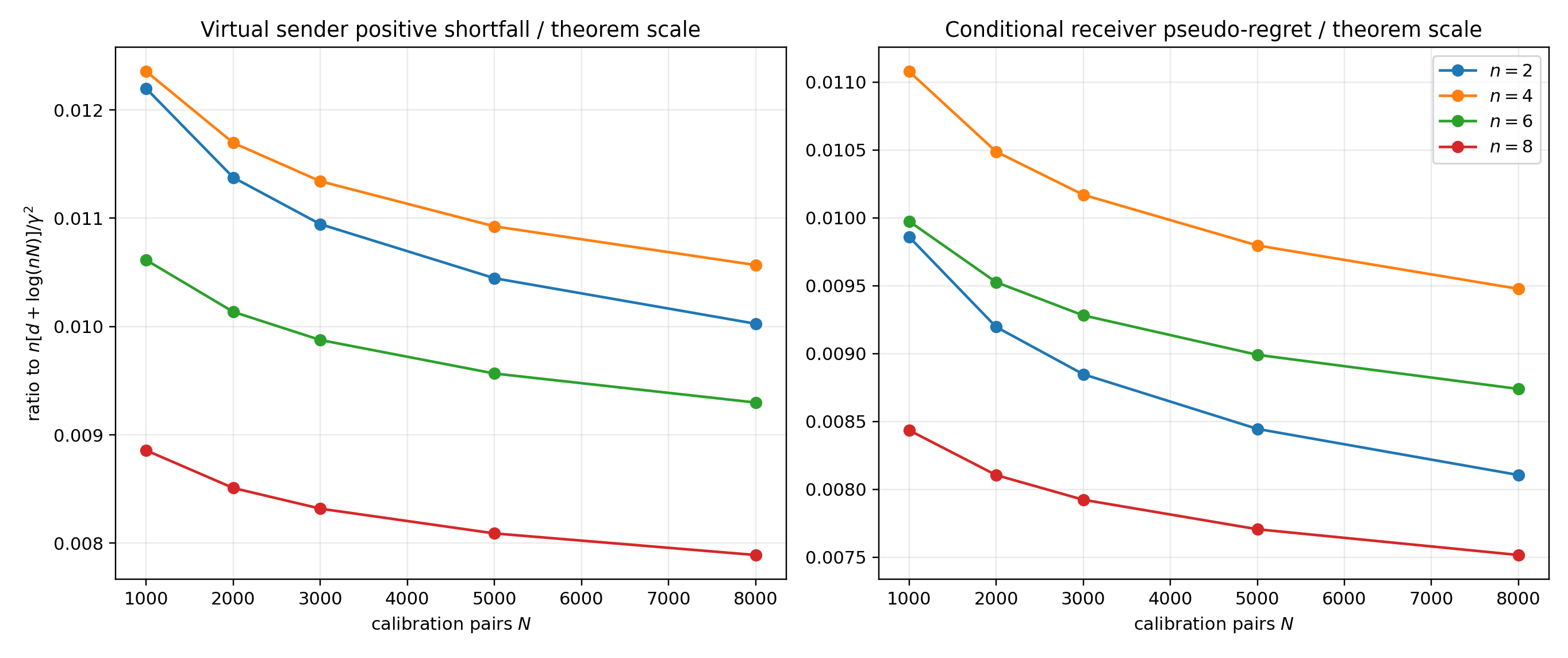}
\caption{Virtual prefix stabilization.  Both cumulative diagnostics are
normalized by the theorem-consistent scale
$n[d+\log(nN)]/\gamma^2$.  The plot illustrates
Proposition~\ref{prop:virtual-regret}; it is not evidence of Bayesian belief
consistency and does not depict a strategic Stage-I equilibrium.  The receiver
curve is conditional (virtual) pseudo-regret: in this symmetric instance, each
incorrect virtual action contributes its conditional expected loss $\gamma$.
It is not cumulative realized utility regret along the sampled state path.}
\label{fig:virtual-stabilization}
\end{figure}

\subsection{Posterior-predictive activation}

The final illustration targets the new equilibrium result.  There are two
equally likely types, two states, matching receiver utility, and common message
support.  The sender also strictly prefers the matching action in each state.
The true type recommends the matching action with probability $0.85$, giving
margin $\gamma=0.7$; the outside type does so with probability $0.35$ and is
not obedient.  Both are $\rho$-implementable at $\rho=0.7$.
For 10,000 certified paths, we compute the exact type posterior and accept only
when both predictive obedience gaps are at least $\gamma/2$.

\begin{figure}[!htbp]
\centering
\includegraphics[width=0.72\textwidth]{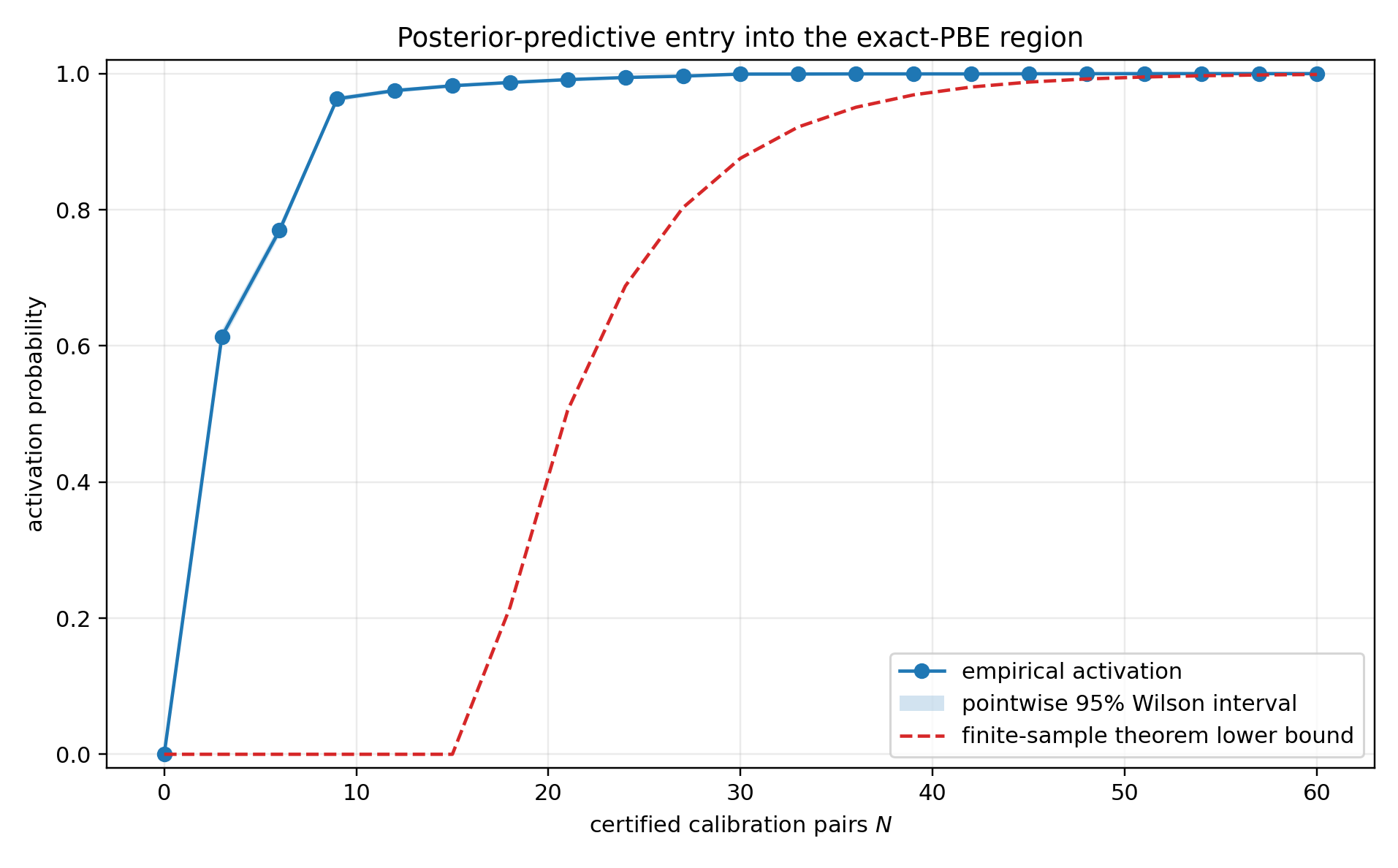}
\caption{Entry into the exact-PBE region.  The dashed curve is
$1-\min\{1,A_E(b^\circ)^N/\alpha^\circ\}$, the direct probability lower bound
from Lemmas~\ref{lem:class-concentration}--\ref{lem:predictive-margin}; the
shaded band is a pointwise 95\% Wilson interval.  Every accepted run satisfies the
deterministic hypotheses of Theorem~\ref{thm:post-calibration-pbe}.  The plotted
frequency is an activation probability, not an equilibrium-approximation
parameter.}
\label{fig:predictive-activation}
\end{figure}

\FloatBarrier

\section{Conclusion}
\label{sec:conclusion}
Opaque commitment creates both a statistical and a strategic problem.
At the population level, the calibration law identifies the receiver-facing
reduced form and the associated equivalence class in the finite type
dictionary.  If that class is a singleton, the type and its stored latent pair
are identified; otherwise only coordinates constant on the class are
identified.  Without additional decomposition restrictions, the same reduced
form may also admit multiple latent representations outside the finite
dictionary.  Finite samples support a predictive posterior and, under
separation, concentrate mass on the identified class.  This predictive result
supports equilibrium only after four
additional steps: beliefs must be constructed on the complete hidden node,
receiver obedience must hold under the posterior mixture of active types,
discretionary recommendations must be sender best, and zero-probability
messages must be completed by beliefs and best responses.

Under these conditions, every accepted record activates an exact PBE of the
fresh deployment continuation.  Finite-type separation, common recommendation
support, and a positive obedience margin make entry into that region likely; the
entry failure probability is not an equilibrium approximation parameter.  The
robust frontier and LP/knapsack results then describe a conditional design
tradeoff between sender value and certifiability.

The result is conditional.  We do not endogenize sender participation,
selection of $x^\theta$, certifier reliability, or policy installation; nor do
we study unrestricted cheap talk at $\rho=0$, equilibrium refinements,
strategic calibration, repeated manipulation, or reputation.  Exact
post-calibration PBE applies to feasible accepted transcripts: $1-\eta$ is an
entry probability, not an $\varepsilon$ equilibrium-approximation guarantee.
Relaxing these institutional or dynamic boundaries requires a different
theory.  Within them, the construction supplies a continuation-equilibrium
foundation for direct following while preserving the tractable receiver-facing
reduced form.

\appendix
\section{Canonical Bayesian Persuasion and an Information Roadmap}
\label{app:information-roadmap}

This appendix places the model on a common set of information coordinates.
The purpose is not to impose our notation on the cited papers, but to separate
three departures from canonical persuasion that are easily conflated:
uncertainty faced by the designer, opacity faced by the receiver, and imperfect
enforceability of the announced or installed policy.

\subsection{Canonical benchmark}
\label{app:canonical-bp}

Let $\Om$ and $\A$ be finite, let $\mu\in\Delta(\Om)$ be a common prior, and
let $u_S,u_R:\A\times\Om\to\mathbb R$ be commonly known.  In canonical
Bayesian persuasion, the sender publicly commits to a signaling scheme
$q:\Om\to\Delta(\M)$ before the state is realized~\cite{KG2011,Kamenica2019}.
After $\omega\sim\mu$ is drawn, the sender observes $\omega$ and a message is
generated according to $m\sim q(\cdot\mid\omega)$.  The receiver observes both
the committed scheme and the realized message, forms
\begin{equation}
 \beta_q(\omega\mid m)
 =\frac{\mu(\omega)q(m\mid\omega)}
 {\sum_{\omega'\in\Om}\mu(\omega')q(m\mid\omega')},
 \label{eq:roadmap-canonical-posterior}
\end{equation}
whenever the denominator is positive, and chooses an action in
\begin{equation}
 \argmax_{a\in\A}\sum_{\omega\in\Om}
 \beta_q(\omega\mid m)u_R(a,\omega).
 \label{eq:roadmap-canonical-best-response}
\end{equation}
The sender selects $q$ to maximize the induced expected value of $u_S$.

Under direct recommendation, one may take $\M=\A$ and interpret message $a$
as the recommendation to choose action $a$.  Writing
$x_{\omega a}=\mu(\omega)q(a\mid\omega)$, Bayes plausibility and obedience are
\begin{align}
 \sum_{a\in\A}x_{\omega a}&=\mu(\omega)
 &&\forall\omega\in\Om,
 \label{eq:roadmap-bayes-plausibility}\\
 \sum_{\omega\in\Om}x_{\omega a}
 \bigl[u_R(a,\omega)-u_R(b,\omega)\bigr]&\ge 0
 &&\forall a,b\in\A.
 \label{eq:roadmap-canonical-obedience}
\end{align}
Thus canonical persuasion selects a state--recommendation distribution
consistent with the prior and receiver optimality.  For the comparison below,
write its comparison tuple as
\begin{equation}
 \mathcal B=(\Om,\A,\M,\mu,u_S,u_R,q).
 \label{eq:roadmap-canonical-environment}
\end{equation}
Before the state is realized, $(\Om,\A,\M,\mu,u_S,u_R)$ is common knowledge,
$q$ is publicly announced and fully binding, and only $\omega$ is hidden from
the receiver.  There is consequently no statistical learning problem in the
one-shot benchmark: the receiver computes \eqref{eq:roadmap-canonical-posterior}
from known objects.

\subsection{Neighboring departures in common information coordinates}
\label{app:information-coordinates}

For any neighboring model, summarize the departure from $\mathcal B$ by
\begin{equation}
 \mathfrak I
 =\bigl(\mathsf K_S,\mathsf K_R;
        \mathsf O_S,\mathsf O_R;
        \mathsf C_q;\mathsf L;\mathsf E\bigr).
 \label{eq:roadmap-information-coordinate}
\end{equation}
Here $\mathsf K_i$ records player $i$'s ex-ante knowledge of the primitive
coordinates, $\mathsf O_i$ the observations accumulated over interaction,
$\mathsf C_q$ who selects and observes the scheme and how strongly it binds,
$\mathsf L$ the object learned from data, and $\mathsf E$ the strategic or
optimization object established.  This is comparison notation only; the cited
papers use their own symbols.  \Cref{tab:roadmap-information,tab:roadmap-target}
give the roadmap, and the paragraphs following them state the distinctions
formally.

\begin{table}[t]
\centering
\caption{Primitive information and commitment relative to canonical Bayesian
persuasion.  ``Public and binding'' means that the receiver observes the
period's scheme and the signal is generated from it.}
\label{tab:roadmap-information}
\footnotesize
\renewcommand{\arraystretch}{1.18}
\begin{tabularx}{\textwidth}{@{}>{\raggedright\arraybackslash}p{0.19\textwidth}YYY@{}}
\toprule
Model or stream & Information missing from sender/designer
& Information missing from receiver & Scheme and commitment \\
\midrule
Canonical BP~\cite{KG2011}
& None in $\mathcal B$ before $\omega$
& Realized $\omega$ only
& $q$ is public and fully binding. \\

Unknown receiver type~\cite{CastiglioniCelliMarchesiGatti2020}
& Current receiver type $k_t$, hence the realized receiver payoff function,
before choosing $q_t$
& No opacity about $q_t$; the receiver knows its own type and the common prior
& Every $q_t$ is public and fully binding. \\

Unknown receiver payoff and prior~\cite{Bacchiocchi2024}
& Both $(\mu,u_R)$; the sender knows its own payoff
& The receiver knows $(\mu,u_R,q_t)$ and does not observe $\omega_t$
& Every $q_t$ is public and fully binding. \\

Unknown prior~\cite{LinLi2025}
& The objective prior $\mu^\star$; payoffs are known
& Each myopic receiver knows $(\mu^\star,u_R,q_t)$
& Every $q_t$ is announced and fully binding. \\

Inferable scheme~\cite{Probine2025}
& The fixed environment and fixed scheme are known to the sender
& The receiver initially does not know the fixed $q$
& $q$ is fixed and fully binding, but opaque to the receiver. \\

Partial commitment~\cite{Min2021}
& None of $(\mu,u_S,u_R,q,\rho)$; the sender knows when it is free to report
strategically
& The receiver does not observe whether the current commitment binds
& The announced $q$ binds with exogenous probability $\rho<1$; reporting in
the nonbinding branch is strategic. \\

This paper
& None at a nonbinding message node: the sender observes realized $\theta$ and
$\omega$ and knows that it can choose a message
& The receiver does not observe the realized
$(\theta,c,u_S^\theta,B^{S,\theta},\pi^\theta,
\overline\sigma^\theta)$
& The type-contingent pair is installed, not chosen in the game; $\pi^\theta$
binds when $c=1$, while the sender chooses when $c=0$. \\
\bottomrule
\end{tabularx}
\end{table}

\begin{table}[t]
\centering
\caption{Dynamic observations, learning targets, and solution objects.  The
table distinguishes sender-side online learning from receiver-side inference
and from the continuation-equilibrium question studied here.}
\label{tab:roadmap-target}
\footnotesize
\renewcommand{\arraystretch}{1.18}
\begin{tabularx}{\textwidth}{@{}>{\raggedright\arraybackslash}p{0.19\textwidth}YYY@{}}
\toprule
Model or stream & Relevant observations & Learning/statistical target
& Strategic or optimization target \\
\midrule
Canonical BP~\cite{KG2011}
& One realized message and then a receiver action
& None
& Sender-optimal persuasive scheme; receiver Bayesian best response. \\

Unknown receiver type~\cite{CastiglioniCelliMarchesiGatti2020}
& Full feedback reveals $k_t$; partial feedback reveals only the receiver's
action
& Performance against an unknown/adversarial sequence of receiver types
& No regret relative to the best signaling scheme in hindsight. \\

Unknown receiver payoff and prior~\cite{Bacchiocchi2024}
& The sender observes receiver action feedback after using $q_t$
& Receiver best-response regions needed to select schemes; separate recovery
of $\mu$ and $u_R$ is unnecessary
& Regret relative to an optimal known-environment signaling scheme. \\

Unknown prior~\cite{LinLi2025}
& Past announced schemes, realized signals, and receiver actions; the learning
algorithm need not use state observations
& $\mu^\star$, inferred through receiver best responses
& Regret relative to the optimal scheme for $\mu^\star$. \\

Inferable scheme~\cite{Probine2025}
& After acting, the receiver observes realized states and updates
state--signal counts
& Signal posteriors induced by the single fixed $q$
& Finite-sample value under receiver inference and design of an inferable
fully committed scheme. \\

Partial commitment~\cite{Min2021}
& The receiver observes the message, but not the binding realization
& No repeated statistical target
& One-shot communication equilibrium with endogenous nonbinding reports. \\

This paper
& Public certified $Z_N=((\omega_t,m_t))_{t=1}^N$; at deployment the receiver
observes $(Z_N,m)$ but not $(\theta,\omega,c)$
& Posterior over the finite type dictionary and the posterior-predictive
reduced form; asymptotically, $x^\theta$ and $[\theta]$
& Exact weak PBE of each accepted deployment continuation; neither Stage-I
regret nor ex-ante protocol selection. \\
\bottomrule
\end{tabularx}
\end{table}
\FloatBarrier

\paragraph{Unknown receiver type or payoff: sender-side online learning.}
In online Bayesian persuasion, the period-$t$ scheme remains visible and fully
enforceable.  In the finite-type formulation of Castiglioni et al.
~\cite{CastiglioniCelliMarchesiGatti2020}, an adversary selects
$k_t\in\mathcal K$, where $k_t$ indexes $u_R^{k_t}$, and the sender selects
$q_t$ without knowing $k_t$.  Full feedback reveals $k_t$ after the round,
whereas partial feedback reveals the receiver action.  The learned object is
therefore not an opaque commitment: the objective is online performance against
the receiver-type sequence.  Bacchiocchi et al.~\cite{Bacchiocchi2024} remove
sender knowledge of
both $\mu$ and $u_R$.  For a slice $z=(q_t(m\mid\omega))_{\omega\in\Om}$,
receiver action feedback reveals regions of the payoff-weighted best-response
map
\begin{equation}
 \operatorname{BR}_{\mu,u_R}(z)
 =\argmax_{a\in\A}\sum_{\omega\in\Om}
 \mu(\omega)z_\omega u_R(a,\omega).
 \label{eq:roadmap-best-response-map}
\end{equation}
Their algorithm can exploit these decision regions without separately
identifying $\mu$ and $u_R$.  In both cases, learning is performed by the
sender to choose a sequence of fully committed schemes, and the global
criterion is regret rather than a PBE of a hidden-discretion continuation.

\paragraph{Unknown prior: designer-side learning with an informed receiver.}
In Lin and Li~\cite{LinLi2025}, the objective state distribution $\mu^\star$
is unknown
to the designer but known to each myopic receiver, while $u_S$ and $u_R$ are
known to the designer.  In each period, the announced $q_t$ is fully binding,
and the receiver best responds using $(\mu^\star,u_R,q_t)$.  Action feedback
locates $\mu^\star$ relative to obedience boundaries of the form
\begin{equation}
 \sum_{\omega\in\Om}\mu^\star(\omega)q_t(m\mid\omega)
 \bigl[u_R(a,\omega)-u_R(b,\omega)\bigr]=0.
 \label{eq:roadmap-unknown-prior-boundary}
\end{equation}
The designer learns the prior in order to approach the optimal
known-$\mu^\star$ commitment.  This is a change in designer information, not a
failure of the receiver to observe the scheme and not partial commitment.

\paragraph{Opaque but fully binding signaling.}
Probine et al.~\cite{Probine2025} keep one fixed scheme $q$ fully binding but
remove the receiver's initial knowledge of it.  After each action, the receiver
sees the realized state and updates state--signal counts (after the stipulated
initialization with one sample from each signal posterior).  If
$N_k(\omega,m)$ is the
count before round $k$, the receiver uses the empirical signal posterior
\begin{equation}
 \widehat\beta_{m,k}(\omega)
 =\frac{N_k(\omega,m)}{\sum_{\omega'}N_k(\omega',m)}
 \label{eq:roadmap-inferable-posterior}
\end{equation}
and best responds to that estimate.  The sender designs a scheme whose value is
robust to this finite-sample inference.  Opacity is therefore receiver-side,
as in our model, but its object is a single fully committed scheme rather than
an aggregate of binding and discretionary sources.

\paragraph{Partial commitment: hidden enforceability and strategic reporting.}
In Min~\cite{Min2021}, the receiver knows the benchmark primitives, the announced
scheme, and the exogenous binding probability $\rho<1$, but does not observe the
realization $c$.  Using the common comparison notation above, a message is generated from
\begin{equation}
 m\sim
 \begin{cases}
 q(\cdot\mid\omega), & c=1,\\
 r_S(\cdot\mid\omega), & c=0,
 \end{cases}
 \qquad c\sim\operatorname{Bernoulli}(\rho),
 \label{eq:roadmap-min-mixture}
\end{equation}
where $r_S$ is determined endogenously in the communication equilibrium.  The
receiver's decision belief is correspondingly a joint posterior on
$(\omega,c)$.  This is a one-shot commitment problem, not a repeated learning
problem.  Related imperfect-commitment models change the sender's manipulation
technology or the credibility constraint in other ways
~\cite{LipnowskiRavidShishkin2022,LinLiu2024}; the common distinction from the
learning papers above is that strategic feasibility of deviations, rather
than estimation of a fully binding scheme, is central.

\paragraph{This paper: receiver-side learning of an opaque partial-commitment
reduced form.}
The common type dictionary is
\[
 \theta\longmapsto
 \bigl(u_S^\theta,B^{S,\theta},x^\theta,
       \pi^\theta,\overline\sigma^\theta\bigr),
 \qquad \theta\in\Tset,
\]
with common prior $\lambda_0$ on the finite set $\Tset$.  The realized
$\theta$ is known to the sender, certifier, and binding device, but not to the
receiver.  Conditional on $\theta$, the certified calibration law is
\begin{equation}
 x^\theta_{\omega m}
 =\mu(\omega)\phi^\theta(m\mid\omega),
 \qquad
 \phi^\theta
 =\rho\pi^\theta+(1-\rho)\overline\sigma^\theta,
 \qquad
 Z_N\sim(x^\theta)^{\otimes N}.
 \label{eq:roadmap-our-reduced-form}
\end{equation}
Calibration updates $\lambda_N(\theta\mid Z_N)$.  At an accepted deployment,
the installed $\pi^\theta$ generates the message if $c=1$, while the sender
chooses a message if $c=0$.  The equilibrium construction verifies that the
installed discretionary prescription $\overline\sigma^\theta$ is optimal once
the receiver follows, and that the receiver follows under the
posterior-predictive mixture.  The target is therefore the exact continuation
PBE conditional on public pre-play data.  It is not a claim that calibration is
equilibrium play, that the sender voluntarily installs the protocol, or that
the paper characterizes all equilibria of an unrestricted partial-commitment
game.

\paragraph{Identification scope: finite dictionaries versus unrestricted
decompositions.}
The population calibration observation law point identifies $x^\theta$.
Within the
paper's finite common-knowledge dictionary it consequently identifies the
equivalence class
\begin{equation}
 [\theta]=\{\vartheta\in\Tset:x^\vartheta=x^\theta\}.
 \label{eq:roadmap-equivalence-class}
\end{equation}
If the dictionary map $\theta\mapsto x^\theta$ is injective, then
$[\theta]=\{\theta\}$ and asymptotic identification of $x^\theta$ also
identifies the type and every dictionary coordinate, including
$(\pi^\theta,\overline\sigma^\theta)$.  The latent pair is unidentified within
the finite model if and only if observationally equivalent types carry
different latent pairs; every other structural coordinate obeys the same
constancy-on-class criterion.

This Bayesian statement is distinct from the nonparametric decomposition
question.  If the candidate dictionary is not fixed, the set
\begin{equation}
 \mathcal D_\rho(\phi;B^S)
 =\left\{(\pi,\sigma):
 \rho\pi+(1-\rho)\sigma=\phi,
 \ \supp\sigma(\cdot\mid\omega)\subseteq B^S_\omega
 \ \forall\omega\right\}
 \label{eq:roadmap-decomposition-set}
\end{equation}
may contain a continuum of decompositions when $0<\rho<1$ and the relevant
support is nondegenerate.  That result shows nonidentification absent further
decomposition restrictions; it does not imply that a given finite dictionary
contains a continuum of candidate types.  Keeping these two statements
separate gives the precise identification boundary used by the paper: the
receiver learns the behaviorally sufficient reduced form (and its dictionary
equivalence class), while recovery of the latent pair depends on whether the
specified dictionary resolves the decomposition.

\section{Institutional Scope and Supporting Comparisons}
\label{sec:discussion}

This appendix collects the institutional interpretation, assumptions, and
crosswalk that support the main results but are not needed for the forward
development of the model.  The limitations that govern the interpretation of
the contribution are summarized in \Cref{sec:conclusion}.

\subsection{Institutional interpretation and comparison with reputation}

\paragraph{Why the two-stage separation is legitimate.}
Equilibrium analysis routinely conditions a continuation game on public
pre-play information.  What matters is that the information-generating process
is correctly specified and does not omit strategic moves whose incentives would
affect its distribution.  Assumption~\ref{ass:certification} meets that
requirement by construction: calibration is generated by a trusted device,
contains no payoff-relevant action, and cannot be manipulated by the sender.
The receiver's posterior after the public record is therefore a common prior
for $\Gamma(z)$, and PBE is imposed on that continuation.

The claim would change if any of these institutional assumptions were removed.
If the sender selected calibration messages, reported its type, hid realized
states, or earned calibration payoffs, the likelihood in
\eqref{eq:type-posterior} would depend on a strategic policy.  One would then
need reporting incentive constraints, history-dependent sender strategies, and
continuation payoffs.  Calling the same data ``training'' would not make those
requirements disappear.

We connect reduced-form inference under hidden partial commitment to exact
post-calibration equilibrium implementation.  Certified calibration updates the
posterior-predictive recommendation beliefs.  A true-type positive margin,
together with posterior concentration and the common-support condition, makes
the predictive gaps positive with high probability; the public gate then
verifies those mixed gaps for the realized record.  Only after that verification
is following strictly optimal.  Sender-best discretionary recommendations close
the remaining PBE incentive loop.

Our $c$ is a transient source shock redrawn in deployment, not a persistent
committed type.  There is no discount factor, recursive continuation value, or
strategic attempt to build a record.  This restriction preserves a linear
receiver-facing reduced form and type-wise LP design.  It also means the paper
does not explain endogenous trust formation.

\subsection{Assumptions doing substantive work}

\begin{enumerate}[leftmargin=1.8em,label=(\roman*)]
\item \emph{Trusted, type-aware infrastructure.}  The certifier and binding
device observe or authenticate $\theta$ without a sender report.  Otherwise
cross-type reporting incentives are missing.
\item \emph{Finite, well-specified common prior.}  The exact posterior and
Hellinger bound use a finite candidate family containing the true reduced form.
Continuous, misspecified, or nonstationary environments require different
learning results.
\item \emph{Common recommendation support.}  This is sufficient for
finite-sample entry into an exact direct-following PBE.  It is not needed for
Theorem~\ref{thm:post-calibration-pbe}, whose certificate checks every
predictive message directly.  It can be replaced by an explicit
exclusive-message safety condition or finite-time elimination.
\item \emph{Sender-best discretionary support.}  This is the source of sender
sequential rationality after following has been proved.  Ties are harmless for
implementation, but can make the latent decomposition less identifiable.
\item \emph{Receiver utility independent of structural type and source.}
This is why the state marginal $\overline\beta_N$ suffices for receiver choice.
If utility depends on $\theta$, $c$, or certification status, the receiver must
use the corresponding marginal of the full $q_N$ and the reduced-form
sufficiency claim changes.
\end{enumerate}

\subsection{Continuity with the reduced-form results}

\Cref{tab:continuity-map} provides an old-to-new crosswalk: it records which
reduced-form, statistical, frontier, and algorithmic objects are retained and
how each is used in the two-stage implementation framework.  The principal new
layer is the composition
\[
\begin{aligned}
\text{type-wise implementability}
&\Rightarrow \text{Bayesian full-node belief}\\
&\Rightarrow \text{predictive obedience}\\
&\Rightarrow \text{sender and receiver sequential rationality}\\
&\Rightarrow \text{post-calibration PBE}.
\end{aligned}
\]
No statistical lemma is used as a substitute for an equilibrium condition.
The reduced-form statistical results enter only through posterior prediction;
the implementation layer separately constructs beliefs on the complete hidden
node, verifies receiver obedience under the posterior mixture, establishes
sender optimality on the discretionary branch, and completes zero-probability
messages.  The robust frontier and LP/knapsack results then retain their
original role as a conditional design tradeoff between sender value and
certifiability.

\section{Supplementary Tables and Crosswalks}
\label{app:supplementary-tables}

\subsection{Notation and observability}

\begin{table}[H]
\centering
\small
\begin{tabularx}{0.96\textwidth}{@{}lYY@{}}
\toprule
Object & Meaning & Observed by receiver at deployment?\\
\midrule
$\theta$ & persistent structural environment type & no; posterior from $z$\\
$\omega$ & fresh payoff-relevant state & no\\
$c$ & fresh binding-source shock & no\\
$m$ & direct recommendation/message & yes\\
$a$ & receiver's final action & chosen by receiver\\
$x^\theta$ & joint receiver-facing reduced form & candidate family known; realization unknown\\
$\pi^\theta,\overline\sigma^\theta$ & latent binding kernel and discretionary prescription & candidate family known; realization and source hidden\\
$\widehat\beta_N$ & empirical state frequency conditional on a message & statistic\\
$\overline\beta_N$ & Bayesian posterior-predictive state belief & decision belief\\
$q_N$ & Bayesian belief on $(\theta,\omega,c)$ & equilibrium belief\\
\bottomrule
\end{tabularx}
\caption{Objects kept distinct throughout the paper.}
\label{tab:notation}
\end{table}

\subsection{Continuity with the reduced-form results}

\begin{table}[H]
\centering
\footnotesize
\renewcommand{\arraystretch}{1.08}
\begin{tabularx}{0.98\textwidth}{@{}>{\raggedright\arraybackslash}p{0.28\textwidth}Y>{\raggedright\arraybackslash}p{0.18\textwidth}@{}}
\toprule
Original object & Role in the two-stage version & Status\\
\midrule
$\rho$-implementable reduced form
& type-wise installed protocol and sender-IC construction
& retained and strengthened\\
Finite-dictionary identification boundary and unrestricted decomposition
multiplicity
& the population law identifies $[\theta]$ exactly and motivates predictive
beliefs when the class is nonsingleton
& clarified\\
Empirical posterior learning
& independent audit and virtual stabilization statistic
& retained, not used as PBE belief\\
Robust obedience margin
& stability radius for receiver sequential rationality
& retained\\
Value--robustness frontier
& design-value versus certification tradeoff
& retained type-wise\\
Support-wise LP
& strict-feasibility and attainment diagnosis
& retained type-wise\\
Binary fractional knapsack
& fixed-type, fixed-margin special oracle
& retained with scope conditions\\
\bottomrule
\end{tabularx}
\caption{Old-to-new result map.}
\label{tab:continuity-map}
\end{table}

\section{Proofs for the Two-stage Construction}
\label{app:two-stage-proofs}

\subsection{Proof of type-wise implementability}

\begin{proof}[Proof of Lemma~\ref{lem:typewise-implementability}]
Necessity follows by summing the mixture outside the sender-best set.  If
$m\notin B_{\omega}^{S,\theta}$, then $\sigma(m\mid\omega)=0$, and hence
\[
\sum_{m\notin B_{\omega}^{S,\theta}}x_{\omega m}
=\rho\mu(\omega)
\sum_{m\notin B_{\omega}^{S,\theta}}\pi(m\mid\omega)
\le\rho\mu(\omega).
\]
The state marginal and nonnegativity follow because $\pi$ and $\sigma$ are
stochastic kernels.

For sufficiency, fix $\omega$ and let
$\phi_m=x_{\omega m}/\mu(\omega)$.  First suppose $0<\rho<1$.  The
implementability inequality is equivalent to
\[
\sum_{m\in B_{\omega}^{S,\theta}}\phi_m\ge1-\rho.
\]
Therefore there exists a vector $y\in\mathbb R_+^\M$ such that
\[
\supp y\subseteq B_{\omega}^{S,\theta},\qquad
0\le y_m\le\phi_m,\qquad
\sum_m y_m=1-\rho.
\]
For example, allocate mass from the coordinates of $\phi$ inside the
sender-best set until total mass $1-\rho$ is reached.  Set
\[
\sigma(m\mid\omega)=\frac{y_m}{1-\rho},
\qquad
\pi(m\mid\omega)=\frac{\phi_m-y_m}{\rho}.
\]
Both vectors are nonnegative, $\sigma$ is supported on the sender-best set,
and
\[
\sum_m\sigma(m\mid\omega)=1,\qquad
\sum_m\pi(m\mid\omega)
=\frac{1-(1-\rho)}{\rho}=1.
\]
They generate $\phi$.

If $\rho=0$, the implementability inequality forces $\phi$ to be supported on
$B_{\omega}^{S,\theta}$; take $\sigma=\phi$ and any kernel $\pi$.  If
$\rho=1$, take $\pi=\phi$ and any sender-best kernel $\sigma$, which exists
because $B_{\omega}^{S,\theta}$ is nonempty.  Applying the construction
statewise completes the proof.
\end{proof}

\subsection{Proof of the static posterior-mixture implementation theorem}
\label{app:static-mixture-proof}

\begin{proof}[Proof of Theorem~\ref{thm:static-mixture-implementation}]
For an on-path $m$, the state marginal of the full Bayesian belief is
$\beta_{\overline x^\lambda}(\cdot\mid m)$.  Dividing
\eqref{eq:mixture-obedience} by $p_m(\overline x^\lambda)$ shows that choosing
$m$ weakly dominates every alternative $a$, and strict inequalities give
uniqueness.

Suppose first that $\rho<1$, and consider an active type $\theta$, state
$\omega$, and a nonbinding sender node.
Every prescribed message
$m\in\supp\overline\sigma^\theta(\cdot\mid\omega)$ has positive aggregate
probability: $\lambda(\theta)>0$, $\mu(\omega)>0$, and $1-\rho>0$ imply that
its contribution to $\overline x^\lambda_{\omega m}$ is positive.  The
receiver therefore follows it.  By
\eqref{eq:sender-best-prescription}, the sender obtains
\[
u_S^\theta(m,\omega)
=\max_{a\in\A}u_S^\theta(a,\omega).
\]
Any message deviation induces some receiver action, possibly randomized.
Every such final action yields at most this statewise maximum, so no deviation
is profitable.  When $\rho=1$, the sender has no message node and this
condition is vacuous.  When $c=1$ for any $\rho$, the binding device, not the
sender, selects the message.

At each on-path receiver information set, the stated belief follows from
Bayes' rule.  For every zero-probability message, choose a probability
distribution on the histories in its receiver information set and let the
receiver select a best response.  This completes
receiver sequential rationality off path.  The sender's prescribed action
already attains the statewise maximum, so this off-path response cannot create
a profitable sender deviation.  Hence strategies are sequentially rational
and beliefs are Bayes consistent wherever Bayes' rule applies.
\end{proof}

\subsection{Bayesian beliefs}

\begin{proof}[Proof of Lemma~\ref{lem:bayesian-beliefs}]
Conditional on $\theta$, certified observations are independent with
probability mass $x^\theta_{\omega_t m_t}$.  The likelihood of $z$ is therefore
$L_\theta(z)$, and Bayes' rule with prior $\lambda_0$ gives
\eqref{eq:type-posterior}.

Now condition on a feasible $z$ and an on-path deployment message $m$.  The
joint conditional probabilities of a hidden node and message are
\begin{align*}
\Pr(\theta,\omega,c=1,m\mid z)
&=\lambda_N(\theta\mid z)\mu(\omega)\rho
\pi^\theta(m\mid\omega),\\
\Pr(\theta,\omega,c=0,m\mid z)
&=\lambda_N(\theta\mid z)\mu(\omega)(1-\rho)
\overline\sigma^\theta(m\mid\omega).
\end{align*}
Summing both expressions over $(\theta,\omega,c)$ gives
\[
\sum_{\theta,\omega}\lambda_N(\theta\mid z)\mu(\omega)
\left[\rho\pi^\theta(m\mid\omega)
{}+(1-\rho)\overline\sigma^\theta(m\mid\omega)\right]
=\overline p_N(m\mid z).
\]
Dividing by this positive probability proves
\eqref{eq:full-belief-binding}--\eqref{eq:full-belief-discretionary}.  Summing
their numerators over $(\theta,c)$ at a fixed $\omega$ gives
$\overline x_N(\omega,m\mid z)$, which proves
\eqref{eq:predictive-state-belief}.
\end{proof}

\subsection{Proof of exact post-calibration implementation}
\label{app:post-calibration-proof}

\begin{proof}[Proof of Theorem~\ref{thm:post-calibration-pbe}]
By Lemma~\ref{lem:bayesian-beliefs}, $q_N$ is Bayes consistent on every
positive-probability receiver information set.  By
\eqref{eq:acceptance-set}, the posterior-mixture reduced form
$\overline x_N(\cdot,\cdot\mid z)$ satisfies
\eqref{eq:mixture-obedience}, with strict conditional slack $\tau$ when
$\tau>0$.  Assumption~\ref{ass:installed} and
\eqref{eq:sender-best-prescription} hold for every active type.  The result
therefore follows from
Theorem~\ref{thm:static-mixture-implementation}.
\end{proof}

\subsection{Equivalence-class concentration}

\begin{proof}[Proof of Lemma~\ref{lem:class-concentration}]
Write $P=x^{\theta^\circ}$.  Every type in $E=[\theta^\circ]$ has likelihood
$P^{\otimes N}(z)$.  Define the class and outside weighted likelihoods
\[
\mathsf A(z)=w_E P^{\otimes N}(z),\qquad
\mathsf B(z)=
\sum_{\vartheta\notin E}
\lambda_0(\vartheta)(x^\vartheta)^{\otimes N}(z).
\]
On every transcript with $P^{\otimes N}(z)>0$,
\[
r_N(z)=\frac{\mathsf B(z)}{\mathsf A(z)+\mathsf B(z)}
\le\min\left\{1,\frac{\mathsf B(z)}{\mathsf A(z)}\right\}
\le\sqrt{\frac{\mathsf B(z)}{\mathsf A(z)}}.
\]
Since the square root is subadditive on nonnegative sums,
\[
r_N(z)
\le
\sum_{\vartheta\notin E}
\sqrt{\frac{\lambda_0(\vartheta)}{w_E}}
\sqrt{
\frac{(x^\vartheta)^{\otimes N}(z)}
{P^{\otimes N}(z)}
}.
\]
Taking expectation under $P^{\otimes N}$ yields
\begin{align*}
\E_{\theta^\circ}r_N
&\le
\sum_{\vartheta\notin E}
\sqrt{\frac{\lambda_0(\vartheta)}{w_E}}
\sum_z
\sqrt{P^{\otimes N}(z)(x^\vartheta)^{\otimes N}(z)}
\\
&=
\sum_{\vartheta\notin E}
\sqrt{\frac{\lambda_0(\vartheta)}{w_E}}
\left(
\sum_{\omega,m}\sqrt{P_{\omega m}x^\vartheta_{\omega m}}
\right)^N
\\
    &\le A_E(b^\circ)^N.
\end{align*}
Markov's inequality gives the probability statement in
\eqref{eq:class-concentration}.  If $E=\Tset$, the posterior outside $E$ is
identically zero.
\end{proof}

\subsection{Predictive convergence and preservation of strict obedience}

\begin{proof}[Proof of Lemma~\ref{lem:predictive-margin}]
Fix a realized transcript and write $r=r_N$.  Every type in $E$ has reduced
form $P$.  If $r>0$, let
\[
Y=\frac{1}{r}
\sum_{\vartheta\notin E}
\lambda_N(\vartheta\mid Z_N)x^\vartheta;
\]
if $r=0$, set $Y=P$.  Then
\[
\overline x_N=(1-r)P+rY.
\]
For $m\in S$, write $p=p_m(P)$ and $q=p_m(Y)$.  The predictive state posterior
is the mixture
\[
\overline\beta_N(\cdot\mid m,Z_N)
=(1-w_m)\beta_P(\cdot\mid m)+w_m\beta_Y(\cdot\mid m),
\]
where the contamination weight satisfies
\[
w_m
=\frac{rq}{(1-r)p+rq}
\le\frac{r}{(1-r)p_{\min}^\circ}.
\]
The $L_1$ distance between two distributions is at most $2$, so
\[
\|\overline\beta_N(\cdot\mid m,Z_N)-\beta_P(\cdot\mid m)\|_1
\le2w_m
\le\frac{2r}{(1-r)p_{\min}^\circ}.
\]
The true class has positive posterior on every transcript generated by $P$, so
$r<1$ almost surely.

For any $a\ne m$, let
$f_{m,a}(\omega)=u_R(m,\omega)-u_R(a,\omega)$.  Its absolute value is at most
$R_R$.  If $P\in\calF_{\rho,\gamma}^{\theta^\circ}$, then
\[
\E_{\beta_P(\cdot\mid m)}f_{m,a}\ge\gamma.
\]
The predictive gap is therefore at least
\[
\gamma
-R_R\frac{2r}{(1-r)p_{\min}^\circ}.
\]
The threshold in \eqref{eq:posterior-threshold} is exactly the solution of
\[
\frac{2R_Rr}{(1-r)p_{\min}^\circ}\le\frac{\gamma}{2}.
\]
Hence every predictive gap is at least $\gamma/2$.  Under common message
support, $S$ is exactly the set of positive posterior-predictive messages, so
$Z_N\in\calC_N(\gamma/2)$.
\end{proof}

\subsection{Proof of high-probability entry}
\label{app:high-probability-proof}

\begin{proof}[Proof of Theorem~\ref{thm:high-probability-entry}]
First suppose $b^\circ\in(0,1)$.  By
Lemma~\ref{lem:class-concentration} and
\eqref{eq:bayesian-sample-complexity},
$r_N\le\alpha^\circ$ with probability at least $1-\eta$.
Lemma~\ref{lem:predictive-margin} then places the transcript in
$\calC_N(\gamma/2)$, and
Theorem~\ref{thm:post-calibration-pbe} supplies the exact continuation PBE.
When $b^\circ=0$ and $N\ge1$, every outside likelihood is zero on the support
of the true observation, so $r_N=0$ almost surely and the same implication
applies.  If $E=\Tset$, $r_N=0$ by definition.
\end{proof}

\section{Identification Details}
\label{app:identification-proofs}

\begin{proof}[Proof of Proposition~\ref{prop:behavioral-sufficiency}]
Conditional on type $\theta$, state $\omega$, and hidden source $c$, the
installed message probabilities are $\pi^\theta(m\mid\omega)$ when $c=1$ and
$\overline\sigma^\theta(m\mid\omega)$ when $c=0$.  Integrating out $c$ gives
$x^\theta_{\omega m}$.  Integrating out $\theta$ under $\lambda$ gives
$\overline x^\lambda_{\omega m}$.  Bayes' rule then yields
\[
\Pr(\omega\mid m;\lambda)
=\frac{\overline x^\lambda_{\omega m}}
{\sum_{\omega'}\overline x^\lambda_{\omega'm}}.
\]
Because receiver utility depends only on the final action and $\omega$, every
receiver expected payoff and best response is determined by this state
posterior.
\end{proof}

\begin{proof}[Proof of Proposition~\ref{prop:dictionary-identification}]
Conditional on $\theta$, each certified observation has distribution
$x^\theta$ on the finite alphabet $\Om\times\M$.  Hence two dictionary types
generate the same population law if and only if they belong to the same class
in \eqref{eq:equivalence-class}.  The identified set is therefore exactly
$[\theta]$, and it is a singleton exactly when the structural type is point
identified.  A functional $g$ is determined by the observation law exactly
when all observationally equivalent types give it the same value.  If the map
$\theta\mapsto x^\theta$ is injective, the unique dictionary entry associated
with the observed law includes the stored latent pair, proving the final
claim.
\end{proof}

\begin{proof}[Proof of Proposition~\ref{prop:latent-nonidentification}]
Fix a state and abbreviate the reduced-form probability vector by
$\phi\in\Delta(\M)$.  Let $K=\supp\phi$ and suppose $|K|\ge2$.  The
nonparametric statement fixes the admissible sender-best set at $B=K$.

A decomposition is equivalent to choosing a binding subdistribution
$r\in\mathbb R_+^K$ satisfying
\[
0\le r_m\le\phi_m,\qquad \sum_{m\in K}r_m=\rho,
\]
and then setting
\[
\pi_m=\frac{r_m}{\rho},\qquad
\overline\sigma_m=\frac{\phi_m-r_m}{1-\rho}.
\]
The point $r=\rho\phi$ is strictly between $0$ and $\phi$ in every coordinate
of $K$.  Choose two distinct coordinates in $K$ and transfer a sufficiently
small amount $\varepsilon$ from one to the other.  This preserves the total
mass and all box inequalities for every $\varepsilon$ in a nondegenerate
interval.  It therefore generates a continuum of distinct feasible
$(\pi,\overline\sigma)$ pairs with the same $\phi$.  These pairs belong to the
unrestricted decomposition space; they need not all appear as entries in the
maintained finite dictionary.  Holding the decompositions at all other states
fixed extends this statewise fiber to a continuum of full kernel pairs.

For the estimation claim, data have the same distribution under
$\theta_0$ and $\theta_1$.  Hence the distribution of any estimator
$\widehat\theta_N$ is also the same.  Pointwise, the triangle inequality gives
\[
d(\theta_0,\theta_1)
\le d(\widehat\theta_N,\theta_0)
{}+d(\widehat\theta_N,\theta_1).
\]
Taking expectation under the common observation law and then using that the
maximum is at least the average proves
\[
\max_{i\in\{0,1\}}\E_{\theta_i}
d(\widehat\theta_N,\theta_i)
\ge\frac12d(\theta_0,\theta_1).
\]
\end{proof}

\section{Empirical Audit and Stabilization Proofs}
\label{app:statistical-proofs}

\begin{proof}[Proof of Proposition~\ref{prop:empirical-audit}]
For a fixed on-path message $m$, its count
$N_N(m)$ is binomial with mean $Np_m(P)\ge Np_{\min}^\circ$.  The
multiplicative Chernoff bound gives
\[
\Pr\left(
N_N(m)<\frac{Np_m(P)}{2}
\right)
\le\exp\left(-\frac{Np_m(P)}{8}\right)
\le\exp\left(-\frac{Np_{\min}^\circ}{8}\right).
\]
Conditional on $N_N(m)=k$, the $k$ states paired with $m$ are independent
draws from $\beta_P(\cdot\mid m)$.  The finite-alphabet inequality of
\citet{Weissman2003} states that
\[
\Pr\left(
\|\widehat\beta_k-\beta_P(\cdot\mid m)\|_1>\varepsilon
\ \middle|\ N_N(m)=k
\right)
\le C_d\exp\left(-\frac{k\varepsilon^2}{2}\right).
\]
On the event $k\ge Np_m(P)/2$, this is at most
\[
C_d
\exp\left(-\frac{Np_{\min}^\circ\varepsilon^2}{4}\right).
\]
A union bound over at most $n$ on-path messages proves
\eqref{eq:empirical-audit-bound}.

For $\varepsilon=\gamma/(2R_R)$, the first expression on the right of
\eqref{eq:empirical-audit-bound} is at most $\eta/2$ under the first sample
condition in \eqref{eq:empirical-sample-complexity}; the second is at most
$\eta/2$ under the second condition.  Finally, for any alternative $a\ne m$,
the difference between the empirical and true payoff gaps is at most
$R_R\varepsilon=\gamma/2$.  A true gap of at least $\gamma$ therefore remains
strictly positive.
\end{proof}

\begin{proof}[Proof of Proposition~\ref{prop:virtual-regret}]
For each on-path $m$, define
\[
K_m
=
\left\lceil
\frac{8R_R^2}{\gamma_m(P)^2}
\log\left(C_dnN^2\right)
\right\rceil.
\]
When $k$ previous occurrences of $m$ are available, the empirical conditional
distribution is formed from those $k$ independent states with law
$\beta_P(\cdot\mid m)$.  The current state has not yet been revealed.  By the
Weissman inequality,
\[
\Pr\left(
\|\widehat\beta_k(\cdot\mid m)-\beta_P(\cdot\mid m)\|_1
>\frac{\gamma_m(P)}{2R_R}
\right)
\le
C_d
\exp\left(
-\frac{k\gamma_m(P)^2}{8R_R^2}
\right).
\]
For $k\ge K_m$, this probability is at most $1/(nN^2)$.  A union bound over
all $m$ and $k\le N$ gives a clean event with probability at least $1-1/N$.
On that event, every virtual response after the $K_m$th occurrence of $m$
equals the recommendation: the empirical payoff gap differs from its true
counterpart by at most $\gamma_m(P)/2$.

The number of errors on the clean event is therefore at most
\[
\sum_{m\in S}\min\{N_N(m),K_m\}.
\]
The map $x\mapsto\min\{x,K_m\}$ is concave, so Jensen's inequality and
$\E N_N(m)=Np_m(P)$ imply
\[
\E\min\{N_N(m),K_m\}
\le\min\{Np_m(P),K_m\}.
\]
The complement of the clean event contributes at most
$N\cdot(1/N)=1$ additional expected error.  Substitution of $K_m$ gives
\eqref{eq:virtual-mistake-bound}, after absorbing ceiling and logarithmic
constants into a universal $C$.

With payoffs in $[0,1]$, a correct virtual action has zero loss relative to
following and an incorrect action contributes at most one to conditional
receiver pseudo-regret.
Likewise,
\[
\left[
u_S^{\theta^\circ}(m_t,\omega_t)
-u_S^{\theta^\circ}(\widehat a_t,\omega_t)
\right]_+
\le\1\{\widehat a_t\ne m_t\}.
\]
Taking sums and expectations gives the two loss bounds.
\end{proof}

\section{Frontier and Algorithmic Proofs}
\label{app:frontier-algorithm-proofs}

\subsection{Positive-margin frontier}

\begin{proof}[Proof of Theorem~\ref{thm:positive-margin-frontier}]
Because $\calF_{\rho,\gamma}^\theta$ expands as $\gamma\to0^+$, nonemptiness of
$\calF_\rho^{\theta,+}$ implies that $V_{\rho,\gamma}^\theta$ is defined throughout some
interval $(0,\overline\gamma]$.  The monotone limit
$L=\lim_{\gamma\to0^+}V_{\rho,\gamma}^\theta$ therefore exists and satisfies
$L\le V_\rho^{\theta,+}$.  Conversely, any
$x\in\calF_\rho^{\theta,+}$ belongs to $\calF_{\rho,\overline\gamma}^\theta$ for some
$\overline\gamma>0$, and therefore to every
$\calF_{\rho,\gamma}^\theta$ with
$0<\gamma\le\overline\gamma$.  Thus $L\ge U_S^\theta(x)$.  Taking the
supremum over $\calF_\rho^{\theta,+}$ gives $L\ge V_\rho^{\theta,+}$.

Equation~\eqref{eq:frontier-decomposition} follows by adding and subtracting
$V_\rho^{\theta,+}$.
\end{proof}

\begin{proof}[Proof of Theorem~\ref{thm:linear-robustification}, upper bound]
Let $\lambda=\gamma/\xi$ and
$x^\gamma=(1-\lambda)x^\star+\lambda y$.  Convexity gives
$x^\gamma\in\calX_\rho^\theta$.  For every $m\in H$ and $a\ne m$,
\[
D_{m,a}(x^\gamma)
=(1-\lambda)D_{m,a}(x^\star)+\lambda D_{m,a}(y)
\ge\lambda\xi=\gamma
\ge\gamma p_m(x^\gamma).
\]
Outside $H$, both component reduced forms have zero message probability.
Hence $x^\gamma\in\calF_{\rho,\gamma}^\theta$.  Since
$u_S^\theta\in[0,1]$,
\begin{align*}
V_\rho^\theta-V_{\rho,\gamma}^\theta
&\le
U_S^\theta(x^\star)-U_S^\theta(x^\gamma)\\
&=\lambda
\left[U_S^\theta(x^\star)-U_S^\theta(y)\right]
\le\frac{\gamma}{\xi}.
\end{align*}
\end{proof}

\begin{proof}[Proof of Theorem~\ref{thm:linear-robustification}, tightness]
Fix $\rho\in(0,1]$.  Let $\Om=\{H,L\}$ with equal prior and
$\A=\M=\{1,0\}$.  The sender obtains $1$ from action $1$ and $0$ from action
$0$ in both states.  The receiver payoff differences are
\[
u_R(1,H)-u_R(0,H)=1,\qquad
u_R(1,L)-u_R(0,L)=-1.
\]
Write
$\alpha_H=\phi(1\mid H)$ and
$\alpha_L=\phi(1\mid L)$.  Because action $0$ is outside the sender-best set,
implementability requires $1-\alpha_\omega\le\rho$.  The sender value is
$(\alpha_H+\alpha_L)/2$.

The weakly obedient choice $\alpha_H=\alpha_L=1$ has value $1$, so
$V_\rho^\theta=1$.  Robust obedience after recommendation $1$ requires
\[
\frac{\alpha_H-\alpha_L}{2}
\ge\gamma\frac{\alpha_H+\alpha_L}{2},
\]
or
$\alpha_L\le[(1-\gamma)/(1+\gamma)]\alpha_H$.  Hence every feasible scheme has
value at most $1/(1+\gamma)$.  The bound is attained by
\[
\alpha_H=1,\qquad
\alpha_L=\frac{1-\gamma}{1+\gamma}.
\]
Its non-sender-best probability in state $L$ is
$2\gamma/(1+\gamma)\le\rho$ exactly when
$\gamma\le\rho/(2-\rho)$.  Recommendation $1$ has the required gap with
equality.  For recommendation $0$,
$D_{0,1}=p_0\ge\gamma p_0$.  Therefore
$V_{\rho,\gamma}^\theta=1/(1+\gamma)$ in the stated range.

Finally, the choices
$\alpha_H=1$, $\alpha_L=1-\varepsilon$ have positive receiver margin and
value approaching $1$ as $\varepsilon\to0^+$.  Hence
$V_\rho^{\theta,+}=1$ and the loss $\gamma/(1+\gamma)$ is linear at zero.
\end{proof}

\subsection{Support-wise LP characterization}

\begin{proof}[Proof of Theorem~\ref{thm:supportwise-characterization}]
Fix $H$ with $\xi_{H,\theta}^{\mathrm{feas}}>0$, and let $y\in K_H^\theta$
be the strict point returned by LP1.  For any $x\in K_H^\theta$ and
$\lambda\in(0,1)$,
\[
x^\lambda=(1-\lambda)x+\lambda y
\]
belongs to $K_H^{\theta,+}$ and converges to $x$ as
$\lambda\to0^+$.  Since $U_S^\theta$ is linear, the supremum over
$K_H^{\theta,+}$ equals the maximum over its closure $K_H^\theta$, namely
$W_H^\theta$.  Taking the maximum over strict-feasible supports gives the
displayed formula for $V_\rho^{\theta,+}$.

If a strict point attains $V_\rho^{\theta,+}$, its support $H$ passes LP1, has
$W_H^\theta=V_\rho^{\theta,+}$, and belongs to the LP2 optimal face, so LP3
returns positive strictness.  Conversely, if a support satisfies the three LP
conditions, the LP3 solution is a strict point with value
$V_\rho^{\theta,+}$ and therefore attains the positive-margin supremum.
\end{proof}

\subsection{Binary-action reduction}

\begin{lemma}[Binary implementability box]
\label{lem:binary-box}
Under the assumptions of \Cref{prop:binary-knapsack},
$x$ is $\rho$-implementable if and only if
$1-\rho\le z_\omega\le1$ for every state.
\end{lemma}

\begin{proof}
The only message outside $B_{\omega}^{S,\theta}=\{1\}$ is $0$.  The
implementability inequality is
$\mu(\omega)(1-z_\omega)\le\rho\mu(\omega)$, which is equivalent to the lower
bound.  The upper bound is the probability constraint.
\end{proof}

\begin{lemma}[Redundancy of recommendation zero]
\label{lem:binary-zero-redundant}
Under \eqref{eq:prior-zero-robust}, robust obedience of recommendation $1$
implies robust obedience of recommendation $0$.
\end{lemma}

\begin{proof}
Obedience of recommendation $1$ is
\[
\sum_\omega\mu(\omega)z_\omega(d_\omega-\gamma)\ge0,
\]
so
\[
\sum_\omega\mu(\omega)z_\omega(-d_\omega-\gamma)
\le-2\gamma\sum_\omega\mu(\omega)z_\omega\le0.
\]
Therefore
\begin{align*}
\sum_\omega\mu(\omega)(1-z_\omega)(-d_\omega-\gamma)
&=
\sum_\omega\mu(\omega)(-d_\omega-\gamma)\\
&\quad-
\sum_\omega\mu(\omega)z_\omega(-d_\omega-\gamma)
\ge0,
\end{align*}
where the final inequality uses \eqref{eq:prior-zero-robust}.
\end{proof}

\begin{proof}[Proof of Proposition~\ref{prop:binary-knapsack}]
For $\omega\in\Om_\gamma^+$, raising $z_\omega$ weakly relaxes
recommendation-$1$ obedience because $d_\omega-\gamma\ge0$ and strictly raises
the sender objective because $s_\omega^\theta>0$.  Hence every optimum sets
$z_\omega=1$ on $\Om_\gamma^+$.

For $\omega\in\Om_\gamma^-$, Lemma~\ref{lem:binary-box} permits the
parameterization
$z_\omega=\underline z_\rho+y_\omega$ with $0\le y_\omega\le\rho$.
Substitution into
\[
\sum_\omega\mu(\omega)z_\omega(d_\omega-\gamma)\ge0
\]
gives
\[
\sum_{\omega\in\Om_\gamma^-}
\mu(\omega)(\gamma-d_\omega)y_\omega
\le C_\gamma^\rho.
\]
If $C_\gamma^\rho<0$, even the lower-box choice $y=0$ is infeasible.  If it is
nonnegative, the sender objective differs from the objective in
\eqref{eq:binary-knapsack} only by a constant.  Recommendation zero is
obedient by Lemma~\ref{lem:binary-zero-redundant}.

For each bad state, the knapsack value and weight are
\[
v_\omega=\mu(\omega)s_\omega^\theta,\qquad
w_\omega=\mu(\omega)(\gamma-d_\omega)>0,
\]
so the value-to-weight ratio is
$s_\omega^\theta/(\gamma-d_\omega)$.  If a lower-ratio state receives positive
capacity while a higher-ratio state is not saturated, transferring a small
amount of weight to the latter strictly raises the objective.  Repeating this
exchange proves the greedy threshold structure and the existence of at most
one fractional state.  Partitioning is linear, sorting costs
$O(|\Om|\log|\Om|)$, and the fill is linear.
\end{proof}

\bibliographystyle{abbrv}
\bibliography{references}

\end{document}